\documentclass[12pt]{article}
\usepackage[utf8]{inputenc}
\usepackage[T1]{fontenc}
\usepackage[margin=1in]{geometry}
\usepackage{amsmath}
\usepackage{amssymb}
\usepackage[dvipsnames]{xcolor}
\usepackage[section]{placeins} % keep floats within section they are written
\usepackage{lineno}
\usepackage[numbers,sort]{natbib}
\usepackage[title]{appendix}
\usepackage{graphicx}
\usepackage{parskip} 
\usepackage{xr}
\usepackage{caption}
\usepackage{subcaption}
\usepackage{float}
\usepackage{amsthm} %to access the proof environment
\usepackage{comment} %to access block comments
\usepackage{ulem} % for access to the sout command
\usepackage{hyperref} %for url to github repo
\usepackage{mathtools} %text within align environment
\usepackage{enumitem}

\newtheorem{lemma}{Lemma}[section]
\newtheorem{theorem}{Theorem}[section]
\newtheorem{definition}{Definition}[section]
\usepackage{titling}
\predate{}
\postdate{}
\date{}

\providecommand{\keywords}[1]{\small \textbf{\textbf{Keywords --}} #1}

\newcommand{\1}{\mathbf{1}} %indicator
\newcommand{\CM}{\mathcal{C_M}}
\newcommand{\R}{\mathbb{R}}
\renewcommand{\P}{\mathbb{P}}

\title{Opinion response functions are key to understanding tipping of social conventions}
\author{Sarah K. Wyse$^{\star}$ and Eric Foxall}

\begin{document}

\maketitle
%\linenumbers

Department of Computer Science, Mathematics, Physics and Statistics, Irving K. Barber Faculty of Science, University of British Columbia Okanagan, SCI354, 1177 Research Road, Kelowna V1V 1V7, BC, Canada\\
$^{\star}$ Corresponding author. \\
\textit{Email addresses}: swyse@student.ubc.ca (S.K. Wyse), eric.foxall@ubc.ca (E. Foxall).\\
\textit{ORCID IDs}: 0000-0002-5500-2711 (S.K. Wyse), 0000-0003-1610-6662 (E. Foxall).

\begin{abstract}
    The extent to which committed minorities can overturn social conventions is an active area of research in the mathematical modelling of opinion dynamics. Researchers generally use simulations of agent-based models (ABMs) to compute approximate values for the minimum committed minority size needed to overturn a social convention. In this manuscript, we expand on previous work by studying an ABM's mean-field behaviour using ordinary differential equation (ODE) models and a new tool, opinion response functions. Using these methods allows for formal analysis of the deterministic model which can provide a theoretical explanation for observed behaviours, e.g., coexistence or overturning of opinions. In particular, opinion response functions are a method of characterizing the equilibria in our social model. Our analysis confirms earlier numerical results and supplements them with a precise formula for computing the minimum committed minority size required to overturn a social convention. 

\end{abstract}

 %these ones from J Stat Mech list - choose 1-4 words

\section{Introduction} \label{Sect:Introduction}
%--------------------------------------------------------------------------------
People are strongly influenced by the opinions and behaviours of those around them and often underestimate the strength of their influence on others \citep{Flynn2008, Bohns2013}. In particular, it has been found that, in specific scenarios, committed minorities making up only 10\% of the total population have the ability to upset a social convention \citep{Xie2011}. Opinion dynamics and the study of tipping points of social conventions are active areas of research in mathematical modelling \citep{Granovetter1978, Galam2007, Centola2018}. Applications of these models can be found in, e.g., linguistics \citep{Baronchelli2006}, vaccination decisions \citep{Papst2022}, and climate action \citep{Constantino2022, Galam2010}.

Individuals in these social models are often assumed to be willing to change their opinions based on decision making processes, e.g., a cost-benefit analysis \citep{Granovetter1978}, following the local majority \citep{Galam2007, Galam2010, Galam2020}, or considering an external influence such as mass media \citep{Andersson2021}. Conversely, an individual who is unwilling to change their opinion is called ``committed'' or ``inflexible''. \citet{Galam2007} state that they are the first to study the impact of adding committed individuals to an opinion dynamics model. For a finite population size, it has since been found that when a committed minority is large enough to reach a critical mass, it can trigger a tipping event \citep{Galam2007, Xie2011, Centola2018, Galam2010, Galam2020}, i.e., an abrupt change in the majority opinion.

Opinion dynamics systems are often studied using an agent-based model (ABM) approach. ABMs consider agents interacting under a set of rules to investigate the behaviour of the overall system. Examples of such models are bounded confidence models \citep{Granovetter1978, Deffuant2000, Cao2015}, the Voter model \citep{Holley1975, Sood2005, Jedrzejewski2018, Majmudar2019}, and the naming game \citep{Baronchelli2006, Xie2011, Centola2018}. Each of these models can be used to study how social conventions are formed through coordination in pairs or small groups. These models may use pairwise interactions in which both individuals update their opinions \citep{Deffuant2000, Cao2015} or they may define a speaking-listening interaction where only the listeners update their opinions \citep{Centola2018}. Other models define interactions as occurring in groups of three where all of the participating agents take on the most frequent opinion in the group \citep{Galam2007, Galam2010}. Results from each of these model types show that consensus of opinion can be achieved, but generally provide little formal analysis of the model dynamics.

One way to gain insight on the driving mechanisms behind model behaviour is employing a well-mixed assumption \citep{Kurtz1976}. Under this assumption, modellers use ordinary differential equations (ODEs) to study the mean-field behaviour of ABMs \citep{Sood2005, Tessone2013, Baumgaertner2018a, Foxall2018}. This approach can allow for quantification of ABM behaviours such as average consensus time, steady states, and bifurcation structure. Some authors suggest that ODE modelling is especially useful when there is limited empirical data since ODE models require fewer assumptions on parameter values as compared to ABMs \citep{vandenBergh2019}. Additionally, a simpler ODE model can provide insights into behaviours of more complex models or determine analytical bounds for plausible ranges of the true behaviour \citep{Wimsatt2007, Smaldino2017}.

In this paper, we aim to fill the gap in formal analysis by using ODE models and a new tool, namely, opinion response functions. Opinion response functions consider a speaking-listening interaction as two separate processes (i.e., a speaking process and a listening process). This decoupling of processes \citep{Diekmann2003} allows us to more easily analyze the equilibria of the system.

The sections in this manuscript are organized as follows. In Section \ref{Sect:Methods}, we describe the ABM social model on which our work is based, its well-mixed ODE approximation, and our new tool, opinion response functions. In Section \ref{Sect:Results}, we present results from each of these methods. Lastly, in Section \ref{Sect:Discussion}, we discuss the resulting model behaviour and possible extensions to this work. An earlier version of this manuscript appears in Chapter 2 of the MSc thesis \citet{Wyse2024}.

\section{Methods} \label{Sect:Methods}
%-------------------------------------------------------------------------------------
\subsection{General Framework} \label{Sect:GeneralFramework}
%-------------------------------------------------------------------------------------
The general framework we adopt for each of our social models is based on the agent-based opinion dynamics model from \citep{Centola2018}. We use the same ABM as the one in that reference. The key terms and definitions are listed in Table \ref{tab:definitions}. We begin with a detailed description of the model.

Individuals in the model can hold one of two opinions, $A$ or $B$. The individuals participate pairwise in one-sided interactions where one individual is randomly assigned as the speaker, making the other individual the listener. Each individual has a memory bank that is a string of memories of length $M$ where each ``memory'' corresponds to one of the two available opinions.

Memory banks are updated as follows. After each interaction, the listener adds the speaker's opinion to their memory bank in a first-in-first-out (FIFO) process. An individual's memory bank is thus a finite history of the opinions they have heard most recently. In other words, each individual's memory bank consists of the opinions spoken in that individual's last $M$ listening interactions. 

The opinion an individual holds and, when applicable, speaks, is the memory that appears most frequently in their memory bank. When $M$ is even, the frequencies of the $A$ and $B$ opinions can be equal. We assume that these undecided individuals speak $A$ and $B$ equally often. The bulk of the population does not inherently favour one opinion over the other, and individuals within this group are willing to update their memory banks and, thus, their opinions. We call this group the uncommitted population and note that individuals in this population may hold either opinion. The remainder of the population is the committed minority and we denote the proportion of individuals in this group by $\mathcal{C_M}\in[0,1]$. Without loss of generality, we set these individuals to hold only $A$ memories in their memory banks. Individuals in the committed minority do not update their memory banks or opinions. We note that we could have equivalently set the committed minority to only hold opinion $B$.

\begin{table}[tbph]
    \centering
    \caption{Key terms used in our social model. }
    \begin{tabular}{lcl}
         Term & Variable & Definition \\
         \hline 
         Memory &  & An entry held by an individual, \\
         & & \qquad can be of type $A$ or $B$ \\
         Memory Bank &  & An individual's collection of memories \\
         Memory Bank Length & $M$ & The number of memories in an \\
         & & \qquad individual's memory bank \\
         Opinion & $O$ & The most frequent memory in an \\ 
         & & \qquad individual's memory bank \\
         Speaking rate & $r_i$ & The speaking rate of opinion $i$ in \\
         & & \qquad the population \\
         Committed Minority & $\CM$ & A small group of individuals who \\
          & & \qquad refuse to update their memory \\
          & & \qquad banks and, thus, their opinions \\
         Social Convention & & The opinion held by the majority \\
         & & \qquad of the population \\
         Consensus & & The entire population holding \\
         & & \qquad the same opinion \\
         \hline
    \end{tabular}
    \label{tab:definitions}
\end{table}

Now, we describe the initial condition and possible resulting behaviour for our model. We aim to investigate the conditions under which this committed minority can overturn a pre-existing social convention, i.e., the opinion held by the majority of the population. As such, we initialize the uncommitted individuals to only hold $B$ in their memory banks which means they also all hold opinion $B$. We say the social convention is overturned when opinion $B$ is no longer held by a majority of the uncommitted population. The minimal critical mass, $\mathcal{C_M^*}$, is the smallest proportion of the committed minority required to cause this tipping event. 

\subsection{Agent-Based Model} \label{Sect:ABM_Methods}
%-------------------------------------------------------------------------------------
Here we provide the details needed to rigorously define and simulate the ABM as given in \citep{Centola2018}. We track the state of each agent using a continuous-time Markov chain (CTMC) given by
\[(X(1,t), X(2,t), \ldots, X(P,t))\] where $X(j,t)$ is the state of agent $j$ at time $t$. The state space for each agent is $S:=\{(x_1,x_2,\ldots,x_M): x_i\in\{A,B\} \ \forall i \}$ and the state space for the entire population is $S^P$ where $P$ is the total number of agents (i.e., the population size). For our simulations, we set $P=10,000$. Let $P_A,P_B$ denote respectively the number of uncommitted agents who hold opinion $A,B$. When $M$ is even, it is possible for some agents to have memory banks that contain an equal number of $A$ and $B$ memories. These agents are undecided; let $P_U$ denote the number of such agents. Undecided agents speak $A$ or $B$ at each time step each with probability $1/2$. Agents $1,...,\lfloor\mathcal{C_M}P\rfloor$ are committed and the other agents are uncommitted. Committed agents are initialized with memory bank $AAA...A$ and never change their opinion. An uncommitted agent with state $(x_1,x_2,\ldots,x_M)$ updates to $(A,x_1,\ldots,x_{M-1})$ at rate $\lfloor\mathcal{C_M}P\rfloor/P + (P_A+0.5P_U)/P$ and updates to $(B,x_1,\ldots,x_{M-1})$ at rate $(P_B+0.5P_U)/P$. 

For each simulation, we choose a memory bank length $M$ and committed minority proportion $\CM$. We note that we run many simulations in which we vary these values so that we can determine the location of $\CM^*$ for $M\in[1,70]$. Simulations of the ABM are run as follows. We use the initial conditions described in Section \ref{Sect:GeneralFramework}. Specifically, the committed minority proportion of the population is $\mathcal{C_M}$ and has memory banks full of opinion $A$, while the proportion of the rest of the population is $1-\mathcal{C_M}$ and has memory banks full of opinion $B$. In each time step, we perform $P/2$ asynchronous interactions so that each agent has, on average, one interaction (speaking or listening) per time step. Each interaction consists of the following steps. First, we randomly choose two agents and set one to be the speaker and the other to be the listener. Second, we determine whether or not the listener is a committed agent and update the listener's memory bank accordingly. Listeners in the committed minority do not update their memory bank, while listeners in the uncommitted population drop their oldest memory and insert the speaker's opinion as the newest entry in their memory bank. Third, and finally, we check whether the new memory changes the listener's opinion and update their opinion accordingly. These three steps are repeated until all $P/2$ interactions are complete. This process is repeated until the simulation reaches a pre-determined end time.

\subsection{Ordinary Differential Equation Model} \label{Sect:ODE_Methods}
%-------------------------------------------------------------------------------------
When an ABM population is large and well-mixed, the mean-field behaviour of the model can be understood using compartmental ODE models in which the rate of movement between compartments is based on the law of mass action \citep{Kurtz1976}. This takes the following form for the above ABM. Let $Y_x$ denote the number of uncommitted individuals with opinion state $x$ and for $O\in \{A,B\}$ let
\begin{align}\label{eq:sigx}
\sigma_O(x_1,\dots,x_M) = (O,x_1,\dots,x_{M-1})
\end{align}
which is the updated state of an individual after hearing opinion $O$. Let $e_x$ denote the unit basis vector in $\R^S$ with $1$ in the $x$ entry and $0$ elsewhere. Summing transition rates over the $Y_x$ uncommitted individuals with opinion state $x$, $Y:=(Y_x)_{x \in S}$ has transition rates
\[q(Y,Y+Z) = \begin{cases}
    (\lfloor\CM P \rfloor + (P_A+0.5P_U))Y_x/P & \text{if} \ Z = e_{\sigma_A(x)}-e_x \\
    (P_B + 0.5 P_U)Y_x/P & \text{if} \ Z = e_{\sigma_B(x)}-e_x \\
    0 & \text{otherwise}.
\end{cases}\]
$Y$ is a Markov chain, since its transition rates depend only on its present value: if $N(x)=\#\{i\colon x_i=A\}$ denotes the number of occurrences of $A$ in opinion state $x$ then
\[P_A = \sum_{x\colon N(x)<M/2}Y_x, \ \ P_U = \sum_{x\colon N(x)=M/2}Y_x \ \ \text{and} \ \ P_B=\sum_{x\colon N(x)>M/2}Y_x.\]
For $y\in \R^S$ define functions $p_A,p_U,p_B$ in the same way, i.e.,
\begin{align}\label{eq:pAy}
p_A(y) = \sum_{x\colon N(x)<M/2}y_x, \ \ p_U(y) = \sum_{x\colon N(x)=M/2}y_x \ \ \text{and} \ \ p_B(y)=\sum_{x\colon N(x)>M/2}y_x.
\end{align}
Define the rate functions $f_Z:\R^S\to \R_+$ by
\[f_Z(y) = \begin{cases}
    (\CM  + p_A(y) + 0.5p_U(y))y_x & \text{if} \ Z = e_{\sigma_A(x)}-e_x \\
    (p_B(y) + 0.5 p_U(y))y_x & \text{if} \ Z = e_{\sigma_B(x)}-e_x \\
    0 & \text{otherwise}.
\end{cases}\]
If $\CM$ is an integer multiple of $1/P$ then transition rates satisfy
\[q(Y,Y+Z) = Pf_Z(Y/P)\]
i.e., $Y$ is a density-dependent Markov chain in the sense of \citep{Kurtz1976}. Defining the vector field
\[F(y) := \sum_Z Zf_Z(y),\]
by Theorem 2.1 in that reference, if $Y(0)/P\to y^{(0)}$ as $P\to\infty$ then $Y(\cdot)/P$ converges uniformly in probability on bounded time intervals to the solution of the initial value problem $y(0)=y^{(0)}$ and $dy/dt = F(y)$. In the present context the ODEs take the following form. Let
\begin{align}\label{eq:rAy}
    r_A(y)=\CM + p_A(y)+0.5p_U(y) \ \ \text{and} \ \ r_B(y) = p_B(y)+0.5p_U(y)
\end{align}
and let $\sigma^{-1}(x_1,\dots,x_M) = \{(x_2,\dots,x_M,O)\colon O \in \{A,B\}\}$. Then
\begin{align}\label{eq:ODEs}
dy_x/dt = -y_x + r_{x_1}(y)\sum_{w \in \sigma^{-1}(x)}y_w
\end{align}
Equation \eqref{eq:ODEs} can be deduced from the above formulas, but the computation is unintuitive. Instead, we shall explain it via inflows and outflows. The loss term $-y_x$ occurs since each individual updates their state at rate $1$, while the other term occurs since each individual with state $(x_2,\dots,x_M,\cdot)$ updates to state $(x_1,\dots,x_M)$ at rate $r_{x_1}(y)$.\\

Recall the population size is $P$, the committed minority has size $\CM P$ and the ODEs are obtained after rescaling by $1/P$. So, the relevant solutions of the ODEs have uncommitted population size $\sum_{x \in S}y_x=1-\CM$. Before moving on we note the following observation that we will need later on.
\begin{lemma}\label{lem:ratesum}
If $\sum_{x \in S}y_x=1-\CM$ then $r_B(y)=1-r_A(y)$.
\end{lemma}
\begin{proof}
    From \eqref{eq:rAy}, $r_A(y)+r_B(y) = \CM + p_A(y)+p_B(y)+p_U(y)$ and from \eqref{eq:pAy}, $p_A(y)+p_U(y)+p_B(y)=\sum_{x \in S}y_x$. If the latter is equal to $1-\CM$ then $r_A(y)+r_B(y)=\CM + 1-\CM=1$.
\end{proof}
 
To provide some context we explicitly give the functions $r_A$ and $r_B$ for $M\in\{1,2,3\}$ as well as for a modified version of the $M=2$ case, in which the uncertain opinion states $\{AB,BA\}$ are treated as a single group. The modified system is introduced to contrast its dynamics with the unmodified $M=2$ case.

\subsubsection{$M=1$ case} \label{Sect:ODE_Methods_M=1}
%-------------------------------------------------------------------------------------
In the case where individuals hold one memory, there are two possible states for the uncommitted population. We call the corresponding compartments $A$ and $B$, where each member of the population holds, respectfully, opinion A or opinion B. When $M=1$, an individual's memory bank and opinion are the same. Individuals only change their opinion when they hear the opposite opinion. That is, the transition rate from $A$ to $B$ is $y_A+\mathcal{C_M}$ and the transition rate from $B$ to $A$ is $y_B$. The movement between groups is shown in Figure  \ref{fig:M=1_Comp}. The system of ODEs is:
\begin{subequations}
\begin{align} \label{eq:Model_M=1}
	\frac{dy_A}{dt} &= y_B(y_A+\mathcal{C_M}) - y_Ay_B = y_B\CM, \\
    \frac{dy_B}{dt} &= y_Ay_B - y_B(y_A+\mathcal{C_M}) = -y_B\CM.
\end{align}
\end{subequations}
There is some cancellation of terms, namely those corresponding to interactions between uncommitted individuals with opinion $A$ and those with opinion $B$. The only surviving term in the ODEs is due to the presence of a committed minority.
\begin{figure}[H]
	\centering
	\includegraphics[width=0.5\linewidth]{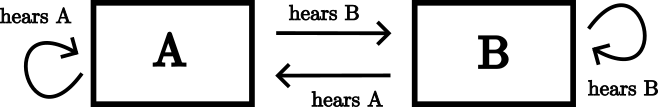}
	\caption{Compartmental diagram for the uncommitted compartments of the $M=1$ ODE model. The compartments that speak opinion $A$ are $A$ and $\mathcal{C_M}$ (not shown) and the compartment that speaks opinion $B$ is $B$.}
	\label{fig:M=1_Comp}
\end{figure}

\subsubsection{$M=2$ case} \label{Sect:ODE_Methods_M=2}
%-------------------------------------------------------------------------------------
In the $M=2$ case of the model, there are four compartments ($AA$, $AB$, $BA$, and $BB$) and thus there are four equations. We note that compartments $AB$ and $BA$ are undecided and speak each opinion equally often. In this case the speaking rate of opinion $A$ (i.e., the ``hears $A$'' transition rate) is $r_A(y)=y_{AA}+\tfrac{1}{2}y_{AB}+\tfrac{1}{2}y_{BA}+\mathcal{C_M}$ and the speaking rate of opinion $B$ (i.e., the ``hears $B$'' transition rate) is $r_B(y)=y_{BB}+\tfrac{1}{2}y_{AB}+\tfrac{1}{2}y_{BA}$. 

We provide a compartmental diagram in Figure \ref{fig:M=2_Comp} that depicts the transitions between compartments. We describe other possible choices for the speaking behaviour of undecided individuals in Section \ref{Sect:Discussion}. The case where these undecided groups are combined is discussed in Section \ref{Sect:ODE_Methods_M=2Combined}.

\begin{figure}[H]
	\centering
	\includegraphics[width=0.37\linewidth]{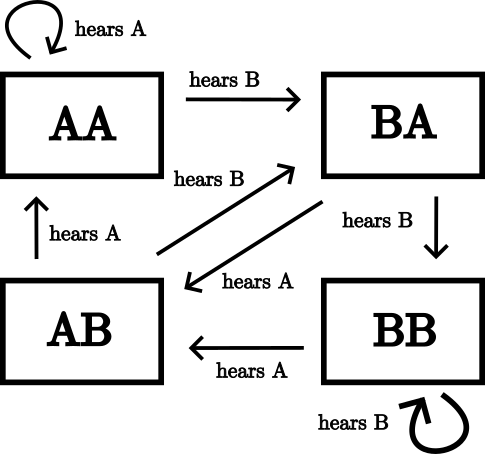}
	\caption{Compartmental diagram for the uncommitted compartments of the $M=2$ ODE model. The compartments that speak only opinion $A$ are $AA$ and $\mathcal{C_M}$ (not shown). The compartment that speaks only opinion $B$ is $BB$. Compartments $AB$ and $BA$ speak both opinions equally often.}
	\label{fig:M=2_Comp}
\end{figure}

\subsubsection{$M=2$ case with combined undecided compartments} \label{Sect:ODE_Methods_M=2Combined}
%-------------------------------------------------------------------------------------
In this variation, which is included here simply to contrast with the usual $M=2$ case, we assume that the $AB$ and $BA$ populations have the same behaviour and combine types $AB$ and $BA$ into an ``undecided'' type, $U$. Moreover, we assume that $U$ switches to whatever it hears next, i.e., to $AA$ if it hears $A$, and to $BB$ if it hears $B$. This process differs from that used in the original $M=2$ model above since, upon hearing $A$, $AB$ transitions to $AA$ but $BA$ transitions to $AB$. This assumption allows us to reduce the previous $M=2$ model to the compartmental diagram given in Figure \ref{fig:M=2_simplified_Comp}. The corresponding system of ODEs is given in Appendix \ref{append:ode_eqs} (Equations \ref{sup:eq:Model_M=2_simp}). Here, the speaking rate of opinion $A$ is $r_A(y)=y_{AA}+\tfrac{1}{2}y_{U}+\mathcal{C_M}$ and the speaking rate of opinion $B$ is $r_B(y)=y_{BB}+\tfrac{1}{2}y_{U}$. The only difference between this simplified $M=2$ model and the original $M=2$ model is the loss of the flow between the $AB$ and $BA$ compartments. The transition rates among the remaining compartments do not change as a result of the simplification since $y_U$ corresponds to $y_{AB}+y_{BA}$.

\begin{figure}[H]
	\centering
	\includegraphics[width=0.8\linewidth]{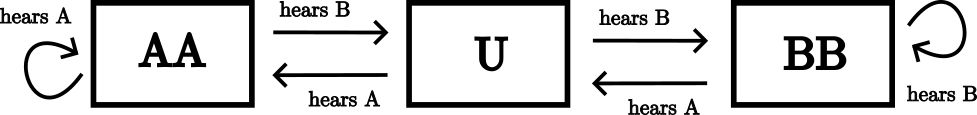}
	\caption{Compartmental diagram for the uncommitted compartments of the $M=2$ with combined undecided compartments ODE model. Compartments $AA$ and $\mathcal{C_M}$ speak only opinion $A$, compartment $BB$ speaks only opinion $B$, and compartment $U$ speaks opinions $A$ and $B$ equally often.}
	\label{fig:M=2_simplified_Comp}
\end{figure}

\subsubsection{$M=3$ Model} \label{Sect:ODE_Methods_M=3}
%-------------------------------------------------------------------------------------
The last ODE model we explicitly consider is the model where individuals can hold 3 memories. A compartmental diagram to show the flow between compartments is given in Figure \ref{fig:M=3_Comp}. In this case $r_A(y)=y_{AAA}+y_{AAB}+y_{ABA}+y_{BAA}+\mathcal{C_M}$ and $r_B(y)=y_{BBB}+y_{BBA}+y_{BAB}+y_{ABB}$.

\begin{figure}[H]
	\centering
	\includegraphics[width=0.85\linewidth]{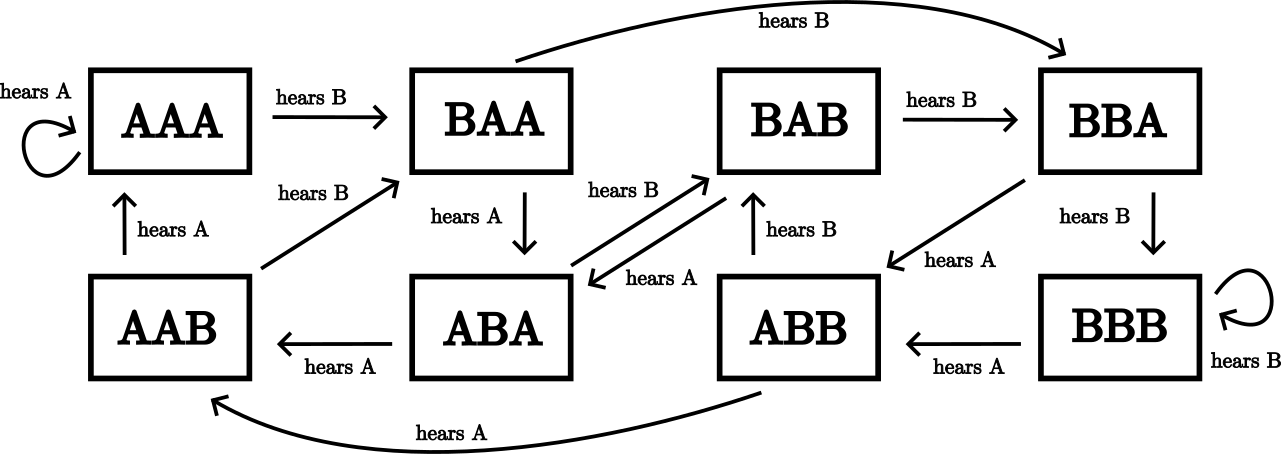}
	\caption{Compartmental diagram for the uncommitted compartments of the $M=3$ ODE model. Compartments $AAA, AAB, ABA, BAA,$ and $\mathcal{C_M}$ speak opinion $A$ and compartments $BBB, BBA, BAB,$ and $ABB$ speak opinion $B$.}
	\label{fig:M=3_Comp}
\end{figure}

The compartment diagram becomes unwieldy for large $M$ since there are $2^M$ opinion states. To study the equilibria of the system for arbitrary $M$, we use opinion response functions, discussed in the next section.

\subsection{Opinion Response Functions} \label{Sect:OpResp_Methods}
%-------------------------------------------------------------------------------------
In this section we characterize the equilibria of the ODE model for any $M$, using a tool we created called the opinion response function (ORF); the ORF for the modified $M=2$ case will be derived in Appendix \ref{append:sig_proof}. We begin with the case $\mathcal{C_M}=0$, i.e., there is no committed minority. To define the ORF for any of the ODE models, we view the opinion dynamics system as a decoupled speaking-listening process, i.e., each interaction is interpreted as a speaking event followed by a separate listening event. When the system is well-mixed and at equilibrium, the speaking rate and the listening rate of each opinion are equal and roughly constant throughout the population and across time. In other words, at equilibrium, the average per-capita rate at which opinion $A$ or $B$ is spoken is, respectively, the same as the rate at which $A$ or $B$ is heard. Overloading notation somewhat, we denote these (constant) listening rates by $r_A,r_B$, satisfying $r_A+r_B=1$. To define the speaking rate, let $X \in S$ denote the opinion state of a fixed individual and recall the opinion update function $\sigma_O$ from \eqref{eq:sigx}. If the listening rate is held constant then $X$ is a Markov chain with transition rates
\begin{align}\label{eq:MC}
x \to \sigma_O(x) \ \ \text{at rate} \ \ r_O \ \ \text{for each} \ \ O \in \{A,B\}.
\end{align}
If $\min(r_A,r_B)>0$, this Markov chain is irreducible: to go from any state to $(x_1,\dots,x_m)$, apply, in order, the transitions $\sigma_{x_m},\dots,\sigma_{x_1}$. The total outgoing rate from a state is $r_A+r_B=1$ so the forward Kolmogorov equations for the transition probabilities are
\begin{align}\label{eq:fwdklmgrv}
\frac{d}{dt}p_t(x,w) = -p_t(x,w) + r_{w_1}\sum_{z \in \sigma^{-1}(w)} p_t(x,z)
\end{align}
where $\sigma^{-1}$ is defined above \eqref{eq:ODEs}. The unique stationary distribution is the product measure $\pi(x) = \prod_{i=1}^m r_{x_i}$. Let $s(x)$ denote the probability that an individual in state $x$ speaks opinion $A$, so that
\begin{align}\label{eq:sx}
    s(x) = \begin{cases}
    0 & \text{if} \ N(x)<M/2, \\
    1/2 & \text{if} \ N(x)=M/2, \\
    1 & \text{if} \ N(x)>M/2.
\end{cases}
\end{align}
The speaking rate of $A$ at equilibrium is then
$\sum_{x \in S} \pi(x) s(x)$. To encode the dependence on listening rate, let $r=r_A$, then $r_B=1-r$ so that the corresponding stationary distribution $\pi_r$ is given by
\begin{align}\label{eq:prodmeas}
\pi_r(x):=r^{N(x)}(1-r)^{M-N(x)}.
\end{align}
Define the corresponding opinion response function $\Phi$ by
\begin{equation}\label{eq:Phi}
    \Phi(r) := \sum_{x \in S} \pi_r(x)s(x).
\end{equation}
That is, $\Phi(r)$ is the probability that a randomly chosen individual speaks $A$, or equivalently the average per-capita speaking rate of $A$, when the per-capita listening rate is $r$, the population is at equilibrium, and there is no committed minority. It can also be written as follows. Let $X(r)$ be a random variable with distribution $\pi$ and let $\1$ denote the indicator function, i.e., $\1(\cdot)=1$ if $\cdot$ is true and $=0$ is $\cdot$ is false. Then $(\mathbf{1}(X_i(r)=A):i=1,\ldots,M)$ are independent Bernoulli$(r)$ trials so $N(X(r))$ is Binomial$(M,r)$ and
\begin{equation}\label{eq:Phi2}
    \Phi(r) = \P(N(X(r))>M/2)+\frac{1}{2}\P(N(X(r))= M/2).
\end{equation}
Let us now compute the ORF when there is a committed minority of size $\CM$ that speaks only $A$. In this case the (rescaled by $1/P$) uncommitted population size is $1-\CM$. If the per-capita listening rate of $A$ is $r$ then the average per-capita speaking rate of $A$ is equal to $\CM + (1-\CM)\Phi(r)$, since the committed population speaks $A$ at per-capita rate $1$ and the uncommitted population speaks $A$ at per-capita rate $\Phi(r)$. This is the more general ORF, that we denote by $\Psi_{\CM}$:
\begin{align}\label{eq:Psi}
    \Psi_{\CM}(r):= \CM + (1-\CM)\Phi(r).
\end{align}
\begin{definition}
    Let $r\in [0,1]$. Then $r$ is a fixed point of $\Psi_{\CM}$ if $r=\Psi_{\CM}(r)$.
\end{definition}

Equilibrium points of \eqref{eq:ODEs} are characterized by the property that listening rate equals speaking rate, i.e., by the fixed point equation $r=\Psi_{\CM}(r)$, in the following sense.

\begin{theorem}\label{thm:ODEeq}
Let $\CM\in [0,1)$. If $y\in \R_+^S$ and $\sum_{x \in S}y_x=1-\CM$ then $y$ is an equilibrium point of \eqref{eq:ODEs} if and only if $y=(1-\CM)\pi_r$ for some $r\in [0,1]$ that satisfies $r=\Psi_{\CM}(r)$.
\end{theorem}

\begin{proof}
Suppose $y$ is an equilibrium point of \eqref{eq:ODEs} that has $\sum_x y_x=1-\CM$, and let $r_A=r_A(y)$ and $r_B=r_B(y)$. By Lemma \ref{lem:ratesum}, $r_B=1-r_A$. Setting $p_t(x,w)=y_w$, the right-hand sides of \eqref{eq:ODEs} and \eqref{eq:fwdklmgrv} are identical, so have the same value, namely $0$. In other words, $y$ is a stationary measure of the Markov chain specified by \eqref{eq:MC}. In particular, letting $r=r_A$, $y=c\pi_r$ for some constant $c$. Since $\sum_x y_x=1-\CM$ and $\sum_x \pi_r(x)=1$, $y=(1-\CM)\pi_r$. From \eqref{eq:pAy}, $p_A(y)+0.5p_U(y)=\sum_{x \in S}y_xs(x)$, so using these and \eqref{eq:rAy},
\begin{align*}
    r=r_A(y)&=\CM + p_A(y)+0.5p_U(y)=\CM + \sum_{x \in S} y_xs(x) \\
    &= \CM + (1-\CM) \sum_{x \in S}\pi_r(x)s(x) = \Psi_{\CM}(r).
\end{align*}
We have verified the necessity of the conditions in Theorem \ref{thm:ODEeq}. If, on the other hand, $y=(1-\CM)\pi_r$ for some $r$ satisfying $r=\Psi_{\CM}(r)$, the above implies $r=r_A(y)$, so with the given choice of $y$, the right-hand side of \eqref{eq:ODEs} is equal to zero, i.e., $y$ is an equilibrium of \eqref{eq:ODEs}.
\end{proof}

For a given value of $\CM$, the form of an equilibrium point $y$ is completely determined by the value of $r$. Moreover, the map $r\mapsto \pi_r$ from values of $r$ to equilibrium points is smooth, in fact analytic, since it is a polynomial in $r$. In other words, Theorem \ref{thm:ODEeq} implies that there is a smooth, one-to-one correspondence between equilibrium points of the ODE system \eqref{eq:ODEs} and fixed points of $\Psi_{\CM}$.

For all $M\ge 3$, $\Phi$ is a sigmoid function, i.e., is bounded, increasing, and has a unique inflection point, see Theorem \ref{thm:Phisig}. $\Psi_{\CM}$ is the weighted average of $1$ and $\Phi$, and increasing $\mathcal{C_M}$ has the effect of ``pulling up'' $\Psi_{\mathcal{C_M}}$ towards $1$. Since $\Phi$ is S-shaped, $\Psi_{\CM}$ has a saddle-node bifurcation in the following sense: there is a value $\mathcal{C_M}^*>0$ such that for $\mathcal{C_M}<\mathcal{C_M}^*$, $\Psi_\mathcal{C_M}$ has three equilibria, for $\mathcal{C_M}=\mathcal{C_M}^*$ it has two equilibria, and for $\mathcal{C_M}>\mathcal{C_M}^*$ it only has the equilibrium $r=1$. This is visible in Figure \ref{fig:OpResp_Func} and is proved in Theorem \ref{thm:sd} in Appendix \ref{append:bifur_proof}.

\begin{figure}[tbph]
    \centering
    \subfloat[$\mathcal{C_M}=0$]{\includegraphics[width=0.48\linewidth]{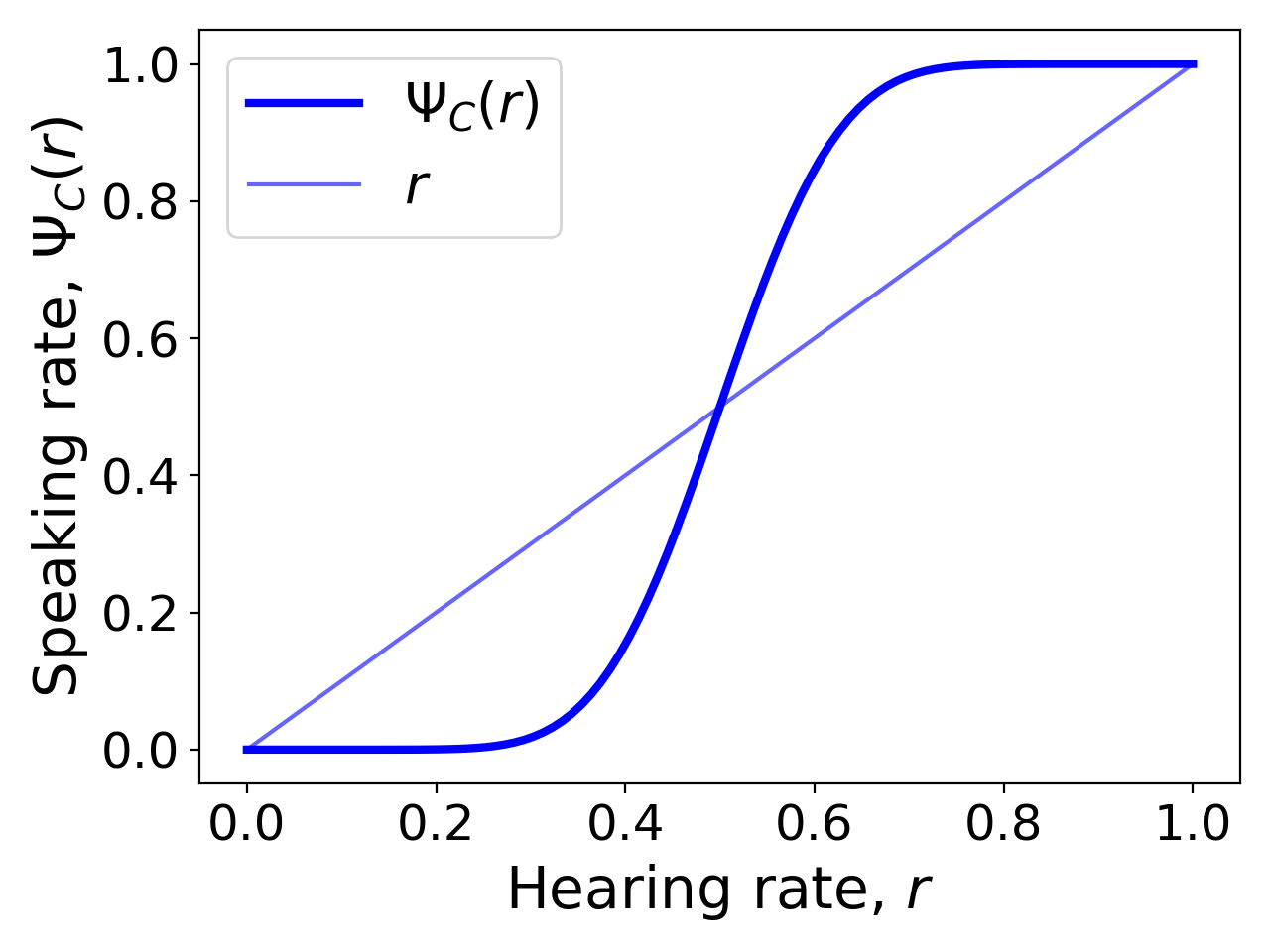}}
    \subfloat[$\mathcal{C_M}=0.31$]{\includegraphics[width=0.48\linewidth]{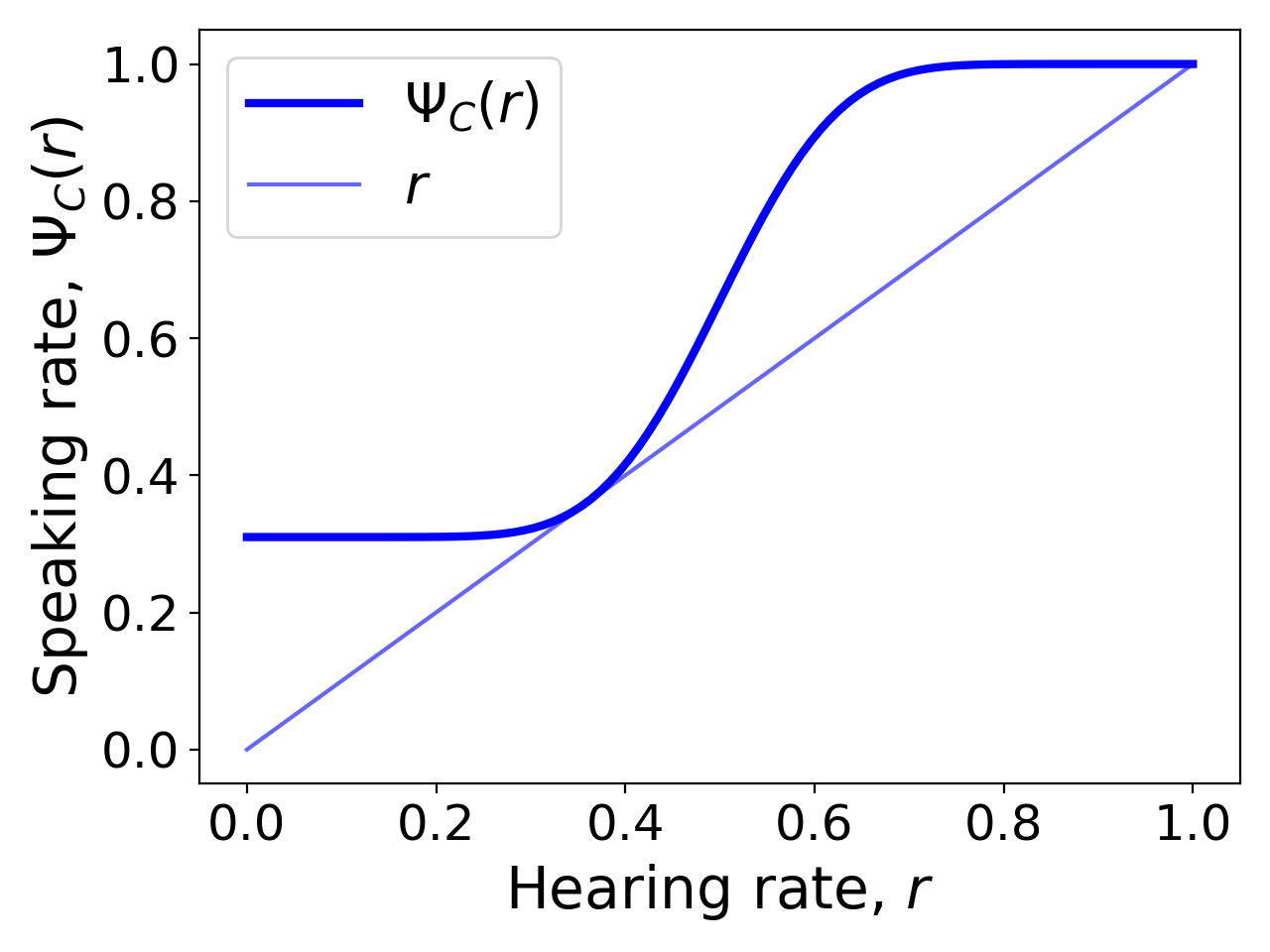}} \\
    \subfloat[$\mathcal{C_M}=0.4$]{\includegraphics[width=0.48\linewidth]{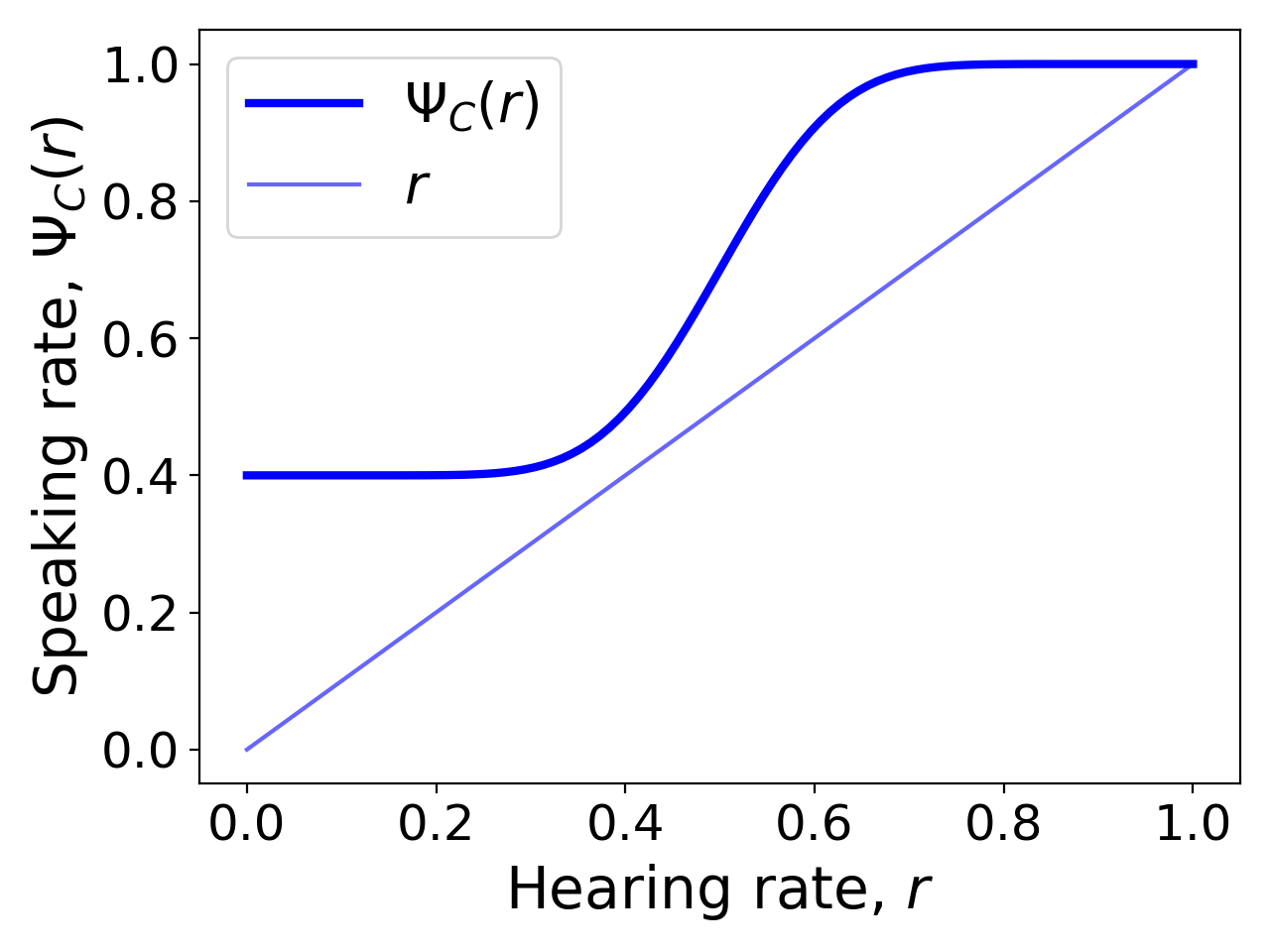}}    
    \caption{Sample opinion response functions, $\Psi_\mathcal{C_M}(r)$, for $M=25$ and the given $\mathcal{C_M}$. We note that $\Psi_\mathcal{C_M}(r) = \Phi(r)$ when $\mathcal{C_M}=0$ and the tipping point for this model occurs at $\mathcal{C_M}=0.31$. An identity line is provided to show the fixed points of $\Psi_\mathcal{C_M}(r)$.}
    \label{fig:OpResp_Func}
\end{figure}

At a practical level, the correspondence between fixed points of $\Psi_{\CM}$ and equilibria of \eqref{eq:ODEs} means that, by solving the fixed point equation $r=\Psi_\mathcal{C_M}(r)$, we can: 1) obtain bifurcation diagrams for any $M$, 2) compute the value of $\mathcal{C_M}^*$, and 3) investigate the limiting behaviour of the system as $M\to\infty$. This is done in Section \ref{Sect:OpResp_Results}, see in particular Figure \ref{fig:OpResp_Bifur}.

\section{Results} \label{Sect:Results}
%-------------------------------------------------------------------------------------
\subsection{Agent-Based Model} \label{Sect:ABM_Results}
%-------------------------------------------------------------------------------------

We simulate the ABM to estimate the tipping point for various $M$ using a binary search algorithm and an error bound of $10/P$ (i.e., the tipping point is found within an accuracy of 10 agents). In order for our ABM results to match those obtained using the ORF approach, we find that the number of agents and timesteps must be sufficiently large. When there are too few agents, we obtain tipping points that occur at a committed minority size smaller than $\mathcal{C_M}^*$. On the other hand, even for $\mathcal{C_M}>\mathcal{C_M}^*$, too few time steps results in the model not having enough time to overturn the social convention. The larger $\mathcal{C_M}$ is, the faster the social convention can be overturned. Moreover, tipping only occurs for committed minority sizes larger than $\mathcal{C_M}^*$. We find that 10,000 agents and 1000 time steps results in tipping points that have an error on the order of 0.001 (10/10,000) when compared to the the tipping points predicted by the opinion response function (see Figure \ref{fig:ABM_TimeSeries} for sample time series). The relationship between memory bank length and committed minority size required to overturn the social convention is increasing and concave except for an observed ``doubling'' of threshold values, see Figure \ref{fig:ABM_TippingPoints}. In Section \ref{Sect:OpResp_Results} we address this doubling by showing that $\P$, when $\CM^*$ are computed from the ORF $\Psi_{\CM}$ given by \eqref{eq:Psi}, $\mathcal{C_{M-\text{1}}}^* = \mathcal{C_M}^*$ for even $M$. We note that opinion $B$ appears to be a passive state in the sense that when $\CM$ is large enough, no one holds opinion $B$ and everyone holds opinion $A$ by the end of the simulation. This phenomenon is a result of the difference in behaviour between the committed minority, who hold opinion $A$ and never update their opinion, and the rest of the population, who are not committed to either opinion and are prone to change their opinion depending on the opinion of their interaction partners.

\begin{figure}[H]
    \centering
    \includegraphics[width=0.48\linewidth]{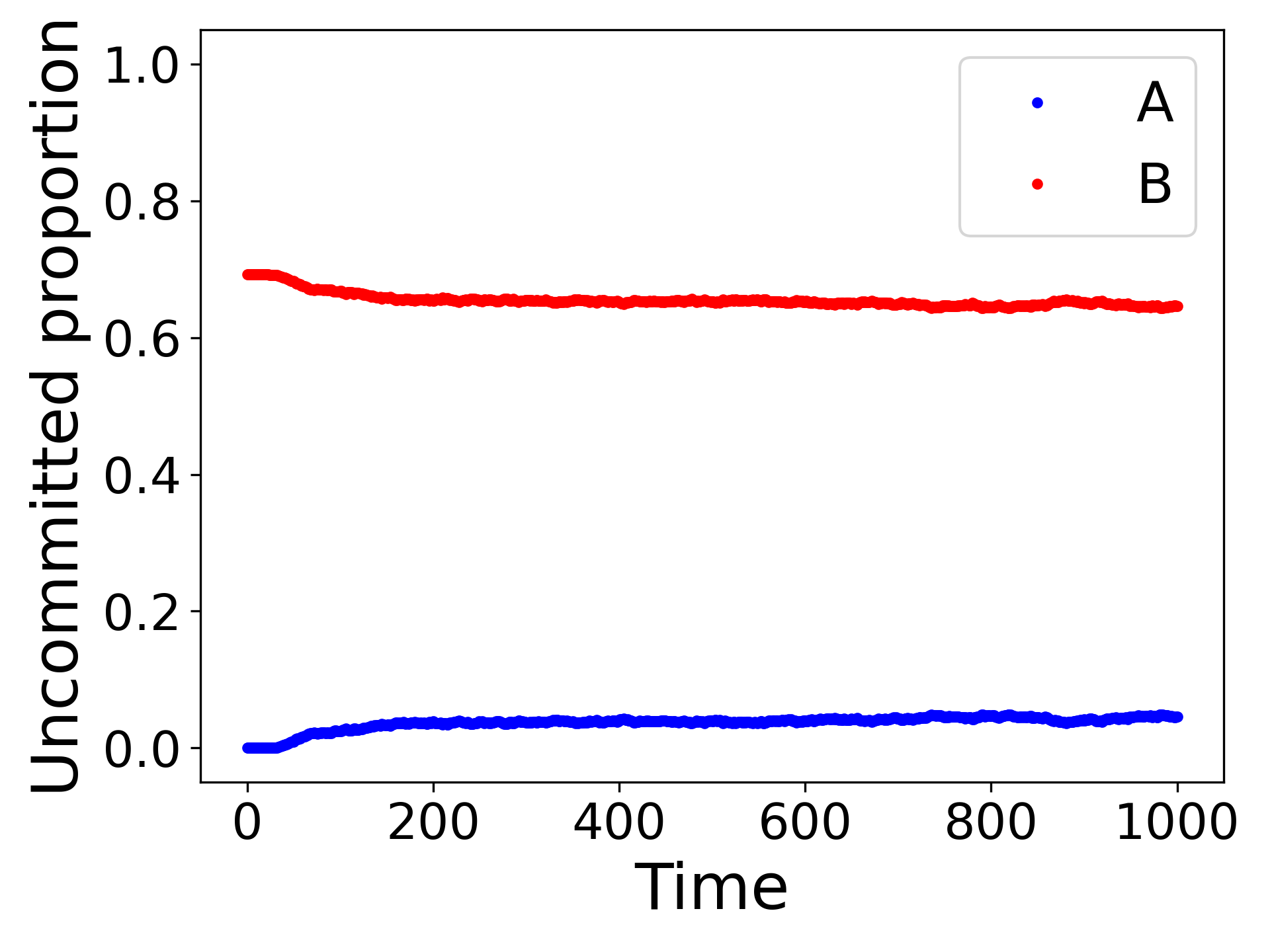}
    \includegraphics[width=0.48\linewidth]{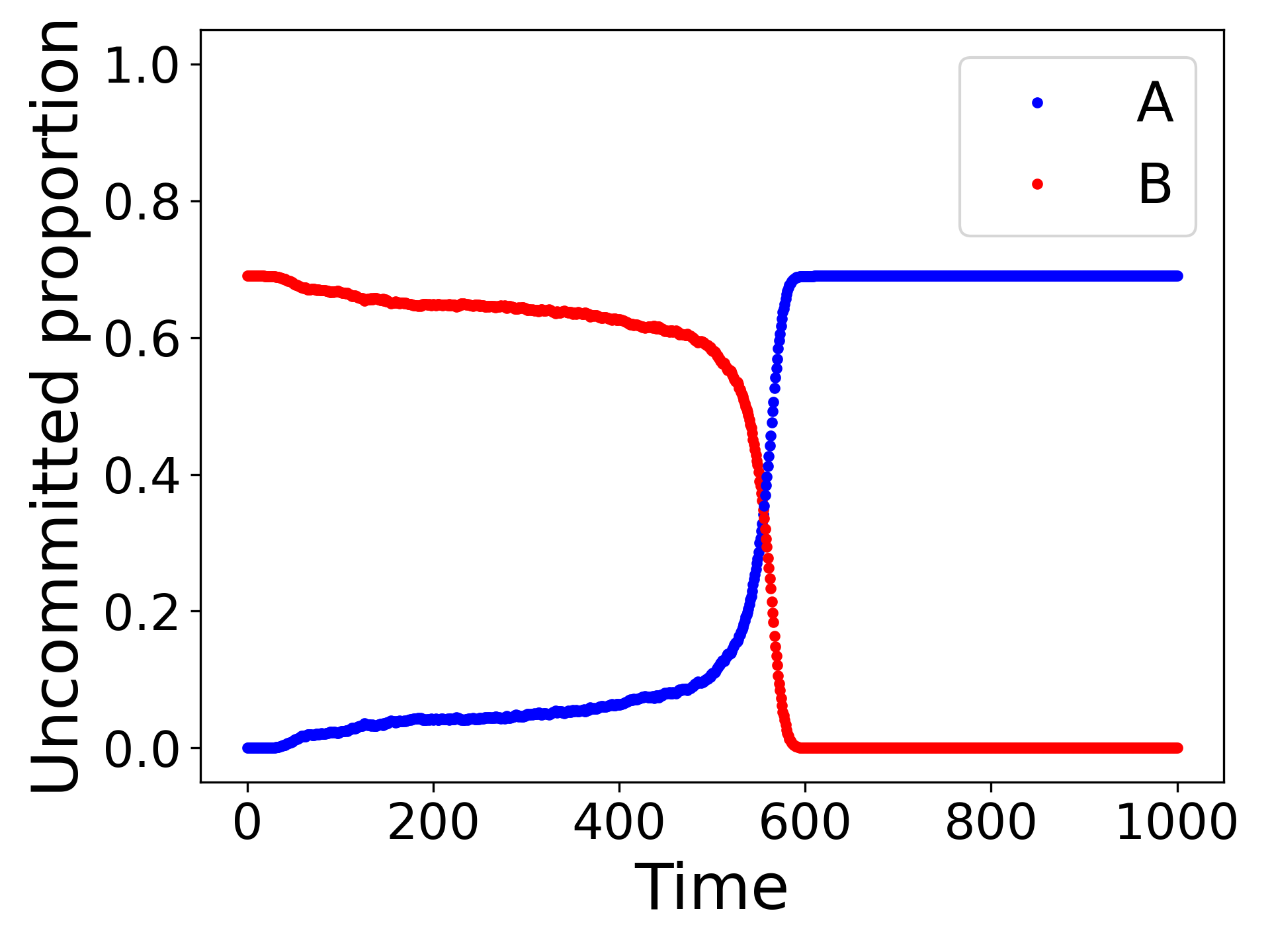}
    \caption{Sample time series for the ABM when $M=25$ and $\CM = 0.308$ (left) and $\CM = 0.31$ (right). Note that $A$ and $B$ are, respectively, the speaking rates of opinion $A$ and $B$ in the uncommitted population.}
    \label{fig:ABM_TimeSeries}
\end{figure}

\begin{figure}[H]
    \centering
    \includegraphics[width=0.6\linewidth]{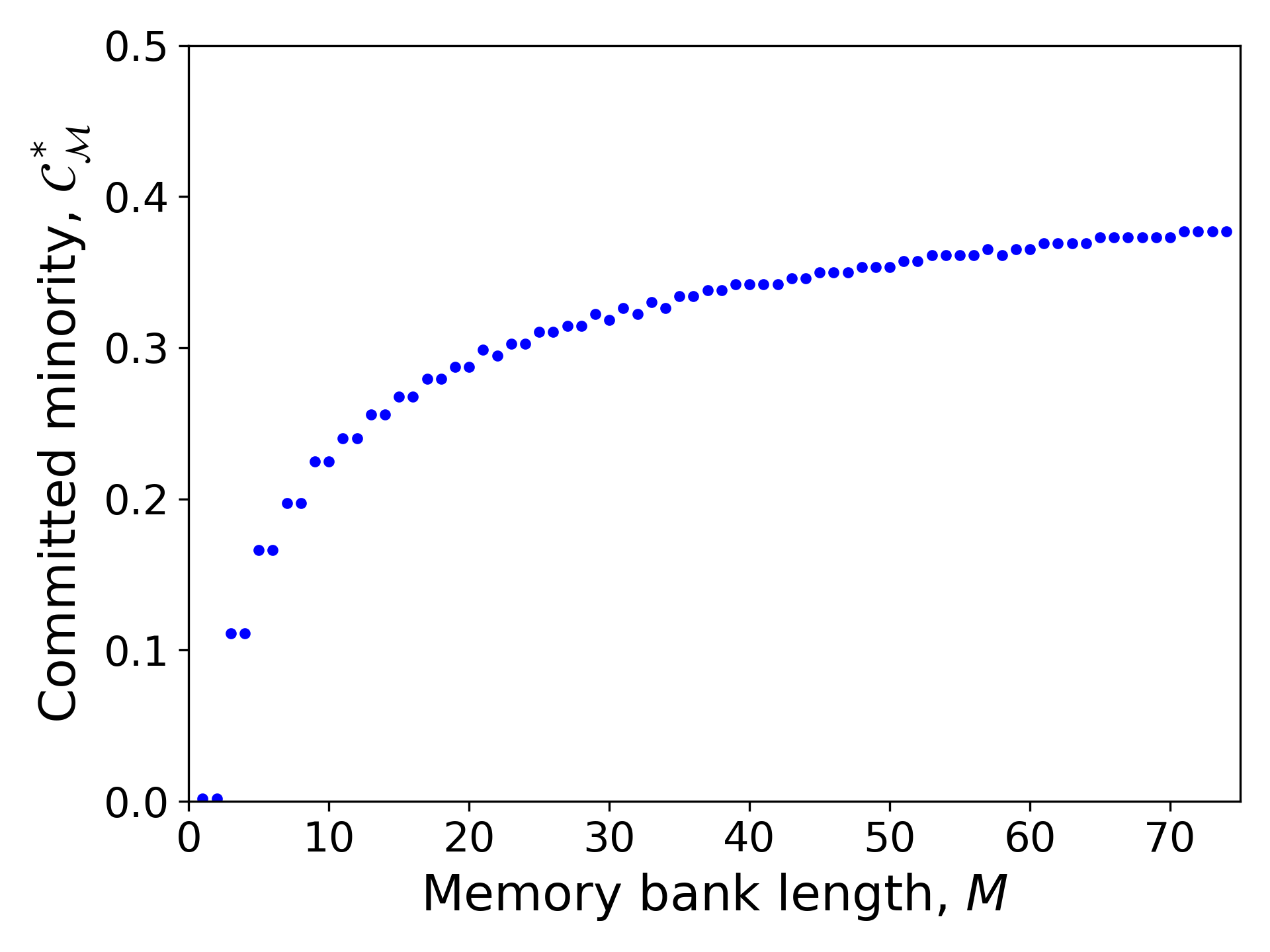}
    \caption{A plot of the estimated $\mathcal{C_M}^*$, the minimal committed minority size required to overturn a social convention, versus $M$, the memory bank length, using the ABM.}
    \label{fig:ABM_TippingPoints}
\end{figure}

\subsection{Ordinary Differential Equation Models} \label{Sect:ODE_Results}
%-------------------------------------------------------------------------------------

\subsubsection{$M=1$}
%-------------------------------------------------------------------------------------
We can simplify the M=1 system of equations by setting the total population to one. After applying the substitution $y_B=1-y_A-\mathcal{C_M}$ and simplifying, we obtain:
\begin{equation}
\label{eq:Model_M=1_simp}
	\frac{dy_A}{dt} = \mathcal{C_M}(1-y_A-\mathcal{C_M}).
\end{equation}
This ODE is always non-negative ($y_A$ is non-decreasing) and is zero when $\mathcal{C_M}=0$ or $1-y_A-\mathcal{C_M}=y_B=0$. When there is no committed minority, the initial state is always stable. We note that the flow rates from $A$ to $B$ and from $B$ to $A$ in Equations \ref{eq:Model_M=1} cancel each other, and there is only a net change if $\mathcal{C_M}>0$. Additionally, when $\mathcal{C_M}=0$, Equation \ref{eq:Model_M=1_simp} is also zero. For a non-zero committed minority, the system always tends towards consensus on opinion $A$ by the entire population. The general solution is: 
\begin{subequations}
\begin{align}
	y_A &= 1-\mathcal{C_M}-D\text{e}^{-\mathcal{C_M}t}, \\
	y_B &= D\text{e}^{-\mathcal{C_M}t},
\end{align}
\end{subequations}
where D is a constant of integration. The $y_B$ population is exponentially decreasing and goes to zero as $t \to \infty$. Overall, we have found that when $\mathcal{C_M}\neq0$ and for any initial conditions, the committed minority succeeds in overturning the social convention. This result means the $M=1$ ODE model only has a trivial tipping point at $\mathcal{C_M}=0$. Time series for this model (Figure \ref{fig:M=1_sim}) show convergence to consensus on $A$ when $\mathcal{C_M}>0$. The only difference between different values of $\mathcal{C_M}\in(0,1]$ is the rate of convergence to consensus. In particular, larger values of $\mathcal{C_M}$ result in faster convergence. 

\subsubsection{$M=2$}
%-------------------------------------------------------------------------------------
Using the initial condition $(y_{AA},y_{AB},y_{BA},y_{BB})=(0,0,0,1-\mathcal{C_M})$ and various values of $\mathcal{C_M}\in[0,1]$, we run time series of the M=2 model. Sample time series are given in Figure \ref{fig:M=2_sim}. Similar to the M=1 ODE model, we find that when $\mathcal{C_M}=0$, there is no movement between compartments and the proportions of the uncommitted population remain the same as in the initial conditions. When $\mathcal{C_M}\in(0,1]$, $y_{BB}$ displays exponential decay. The undecided populations $AB$ and $BA$ initially increase with $AB$ having a slightly larger proportion throughout the simulation. Then, these undecided populations transition to state $AA$ which increases until $AA$ makes up the entirety of the uncommitted population. As $\mathcal{C_M}$ is increased, the rate of convergence increases and the the degree of centering decreases, i.e., the undecided populations reach a smaller maximum.

\begin{figure}[tbph]
    \centering
    \subfloat[$M=1$ model for $\mathcal{C_M}=0$ (left) and $\mathcal{C_M}=0.05$ (right).\label{fig:M=1_sim}]{\includegraphics[width=0.48\linewidth]{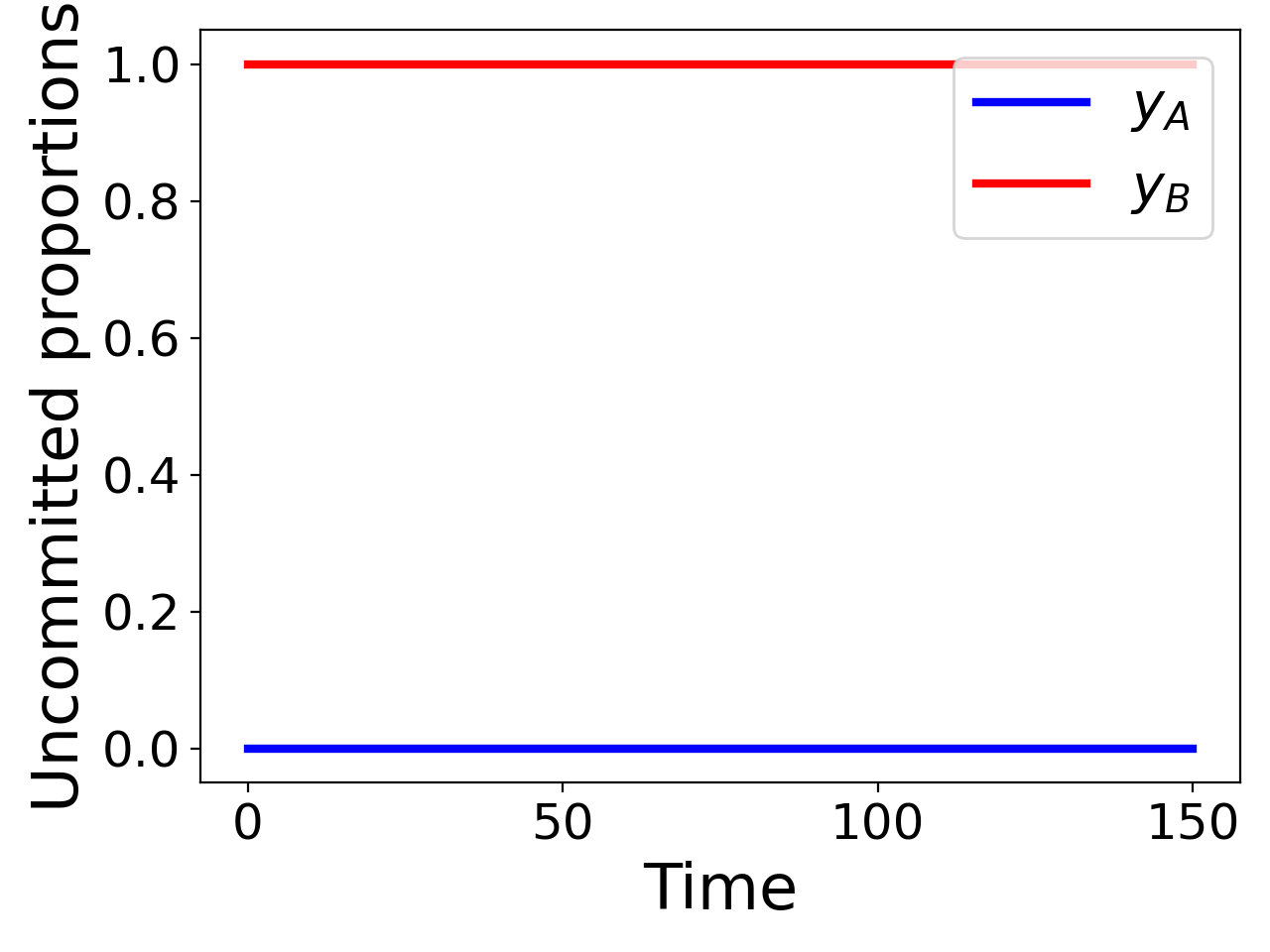} \includegraphics[width=0.48\linewidth]{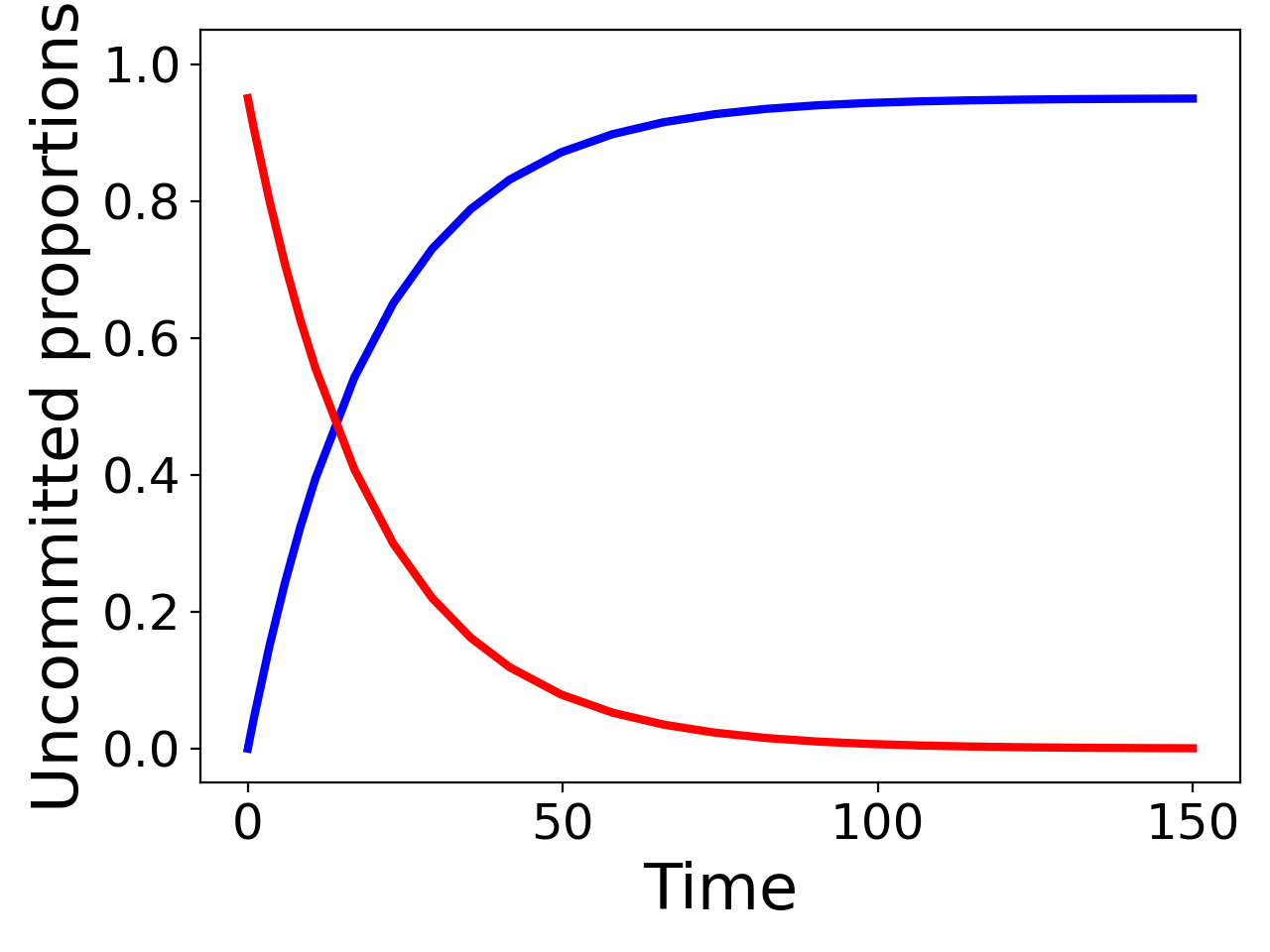}} \\

    \subfloat[$M=2$ model for $\mathcal{C_M}=0$ (left) and $\mathcal{C_M}=0.05$ (right). Note that $y_{AA}, y_{AB}$, and $y_{BA}$ all overlap when $\mathcal{C_M}=0$. \label{fig:M=2_sim}]{\includegraphics[width=0.48\linewidth]{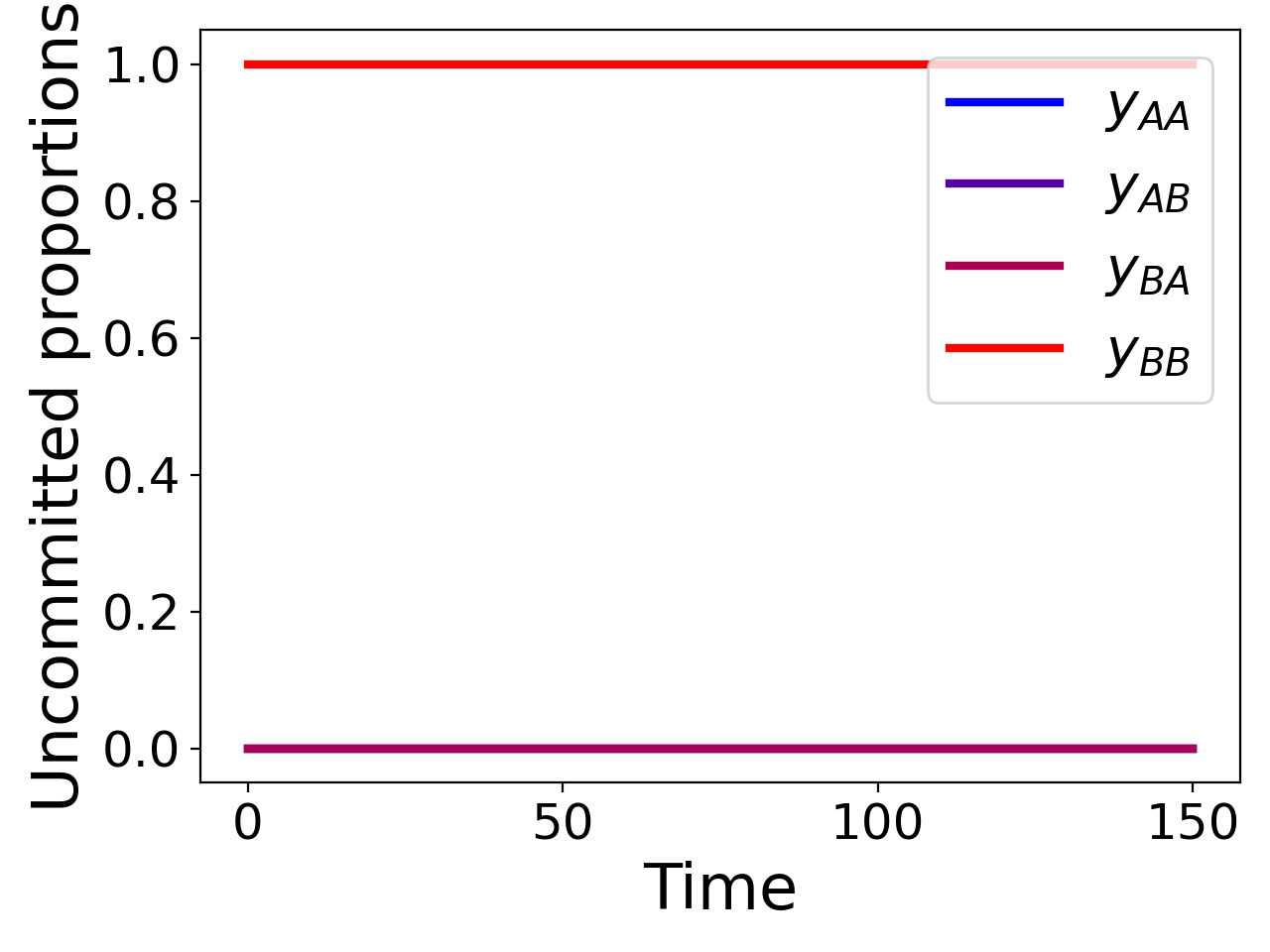} \includegraphics[width=0.48\linewidth]{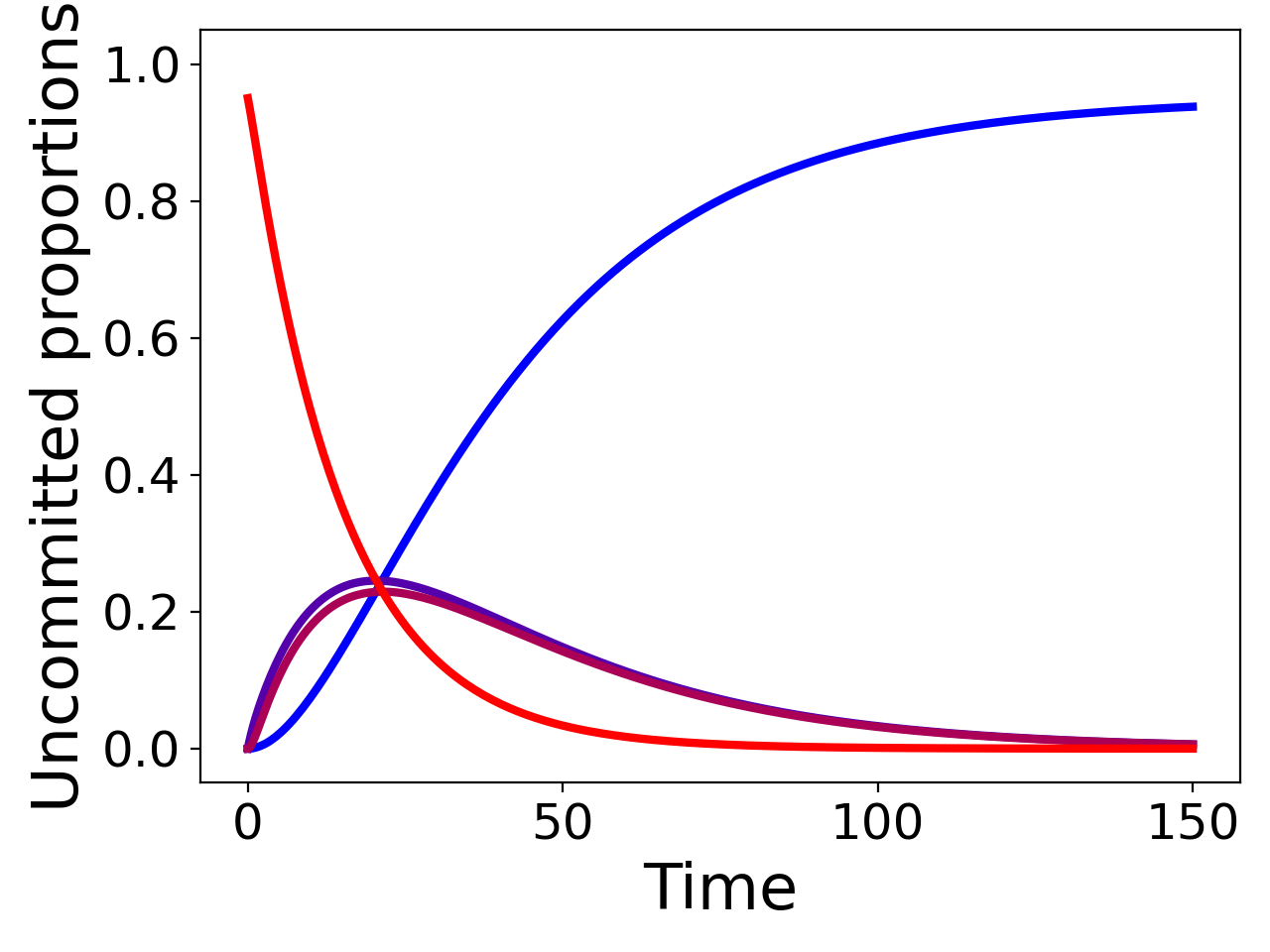}}

	\caption{Sample time series for the ODE models with a trivial tipping point. Note that the legends are the same within each row.}
\end{figure}

\subsubsection{$M=2$ with Combined Undecided Compartments}
%-------------------------------------------------------------------------------------
We use the initial condition $(y_{AA},y_{U},y_{BB})=(0,0,1-\mathcal{C_M})$ to run time series for the M=2 model with combined undecided compartments. Sample time series are given in Figure \ref{fig:M=2_simplified_sim}. The results of this model differ from the previous two ODE models since this model has a non-trivial tipping point. When $\mathcal{C_M}$ is small, the social convention is not overturned and the $BB$ population remains dominant with a small increase in $y_U$ and an even smaller increase in $y_{AA}$. However, when $\mathcal{C_M}$ is large enough, there is a sharper increase in $y_U$ which increases the speaking rate of $A$ and provides enough force for the $AA$ population to eventually overturn the social convention. Similar to the previous $M=2$ case, increasing $\mathcal{C_M}$ beyond the tipping point increases the rate of convergence to consensus on $A$ and decreases centering behaviour.

\subsubsection{$M=3$}
%-------------------------------------------------------------------------------------
Once again, we use the initial condition where the uncommitted population consists of only $BBB$ and all the other uncommitted compartments are empty. Sample time series are given in Figure \ref{fig:M=3_sim}. Similar to the $M=2$ combined undecided compartments model, there is a non-trivial tipping point. When $\mathcal{C_M}$ is less than the tipping point the solution trajectories go to a coexistence equilibrium, with most of the uncommitted population still holding opinion B. The committed minority is not large enough to overturn the social convention and only a small proportion of the population holds opinion $A$. On the other hand, when $\mathcal{C_M}$ is larger than the tipping point, the compartments with two $B$ memories (e.g., $ABB$) increase which causes an increase in the compartments with two $A$ memories (e.g., $ABA$) which causes the $AAA$ population to increase until all of the uncommitted population is in the $AAA$ compartment and there is a consensus on opinion $A$ throughout the population. 

Note that the overturning of the social convention appears to be slightly faster in the $M=3$ model than in the $M=2$ model with combined undecided compartments. This difference in convergence rate occurs because the value of $\mathcal{C_M}$ used for the $M=2$ model with combined undecided compartments is closer to the tipping point. In particular, values of $\mathcal{C_M}$ just above the tipping point have less force pushing the overturning of the social convention than values of $\mathcal{C_M}$ much greater than the tipping point. 

\begin{figure}[tbph]
    \centering
    \subfloat[$M=2$ with combined undecided compartments model for $\mathcal{C_M}=0.07$ (left) and $\mathcal{C_M}=0.08$ (right). \label{fig:M=2_simplified_sim}]{\includegraphics[width=0.48\linewidth]{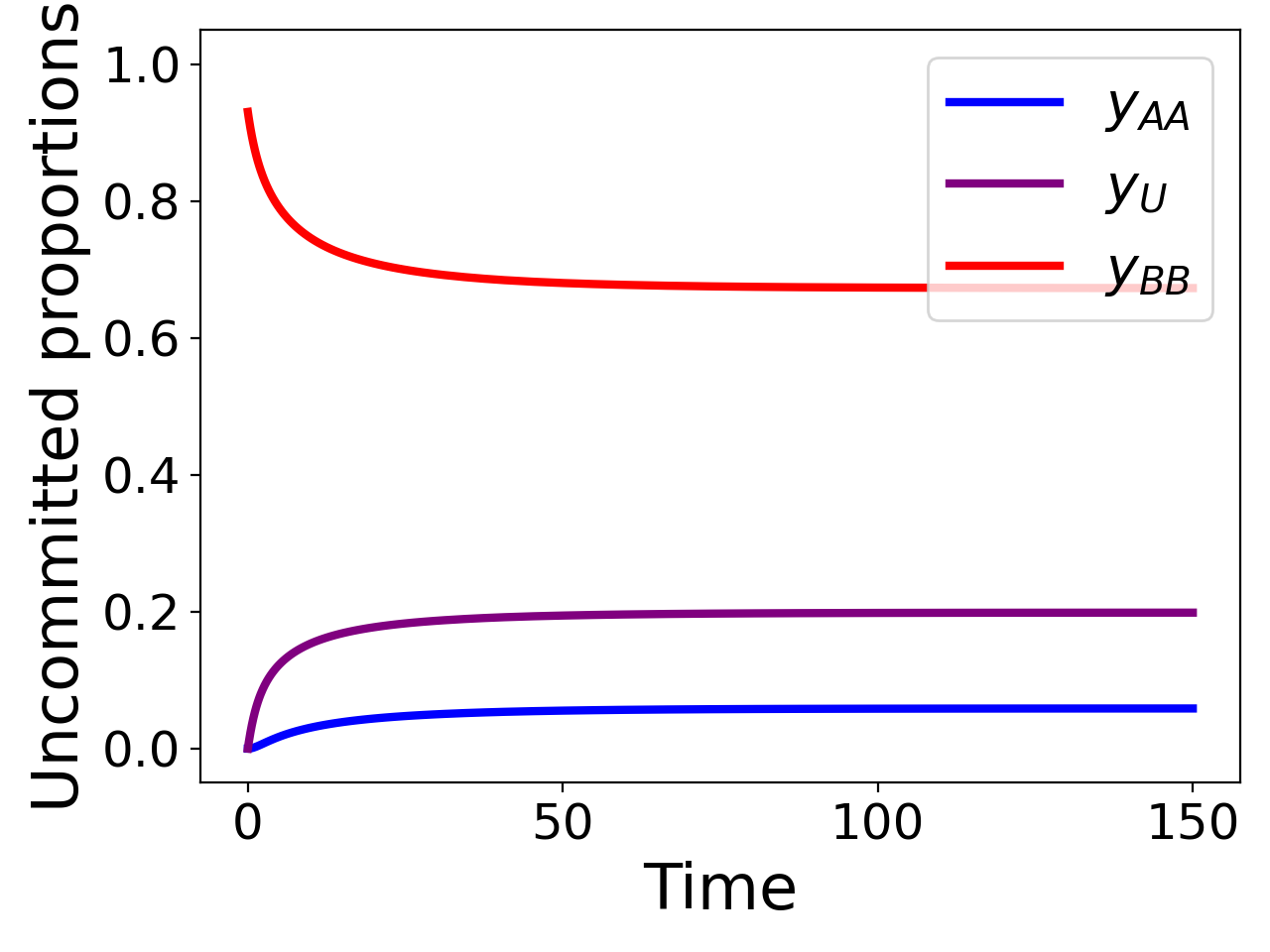} \includegraphics[width=0.48\linewidth]{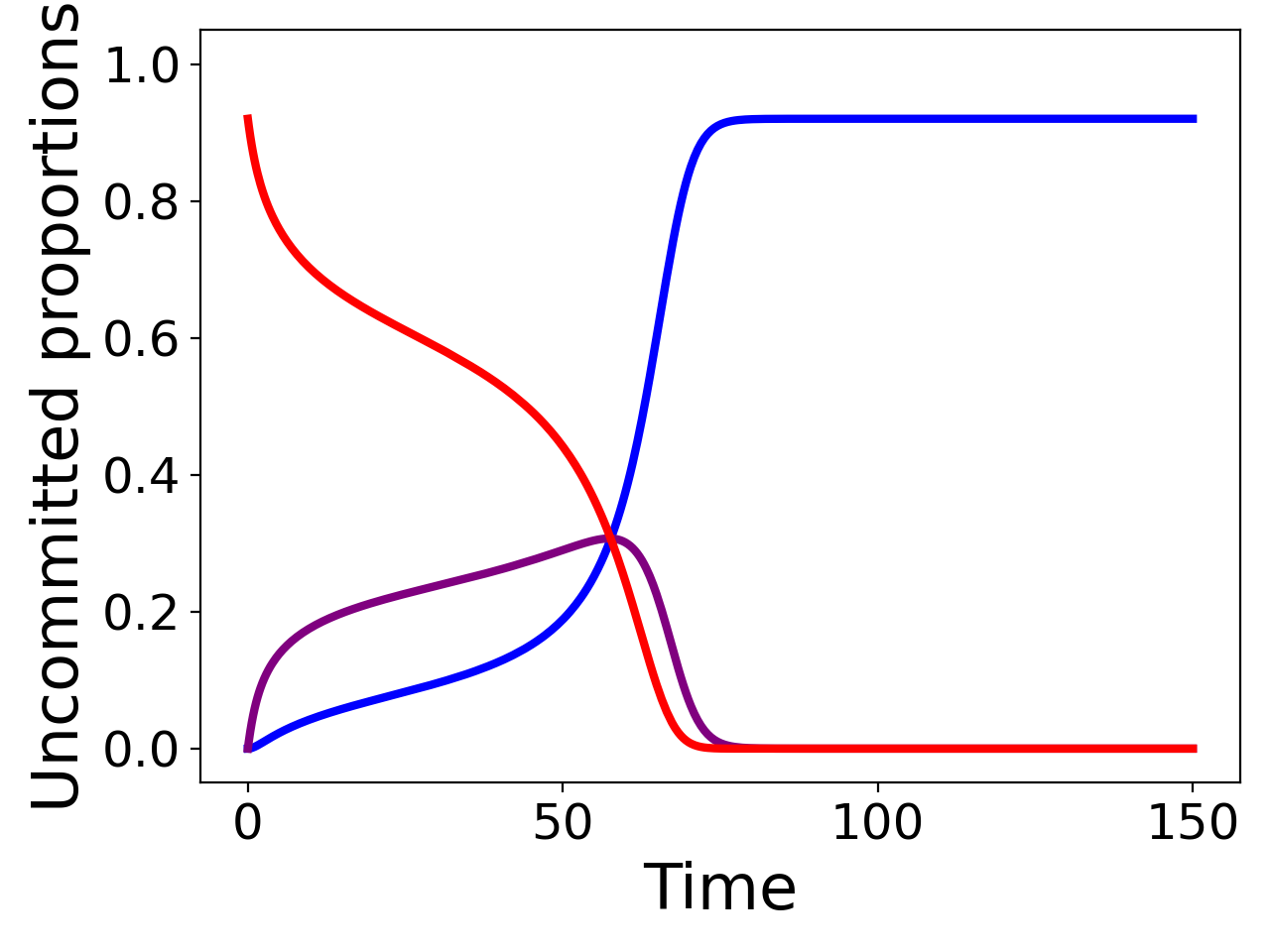}} \\

    \subfloat[$M=3$ model for $\mathcal{C_M}=0.11$ (left) and $\mathcal{C_M}=0.12$ (right). \label{fig:M=3_sim}]{\includegraphics[width=0.48\linewidth]{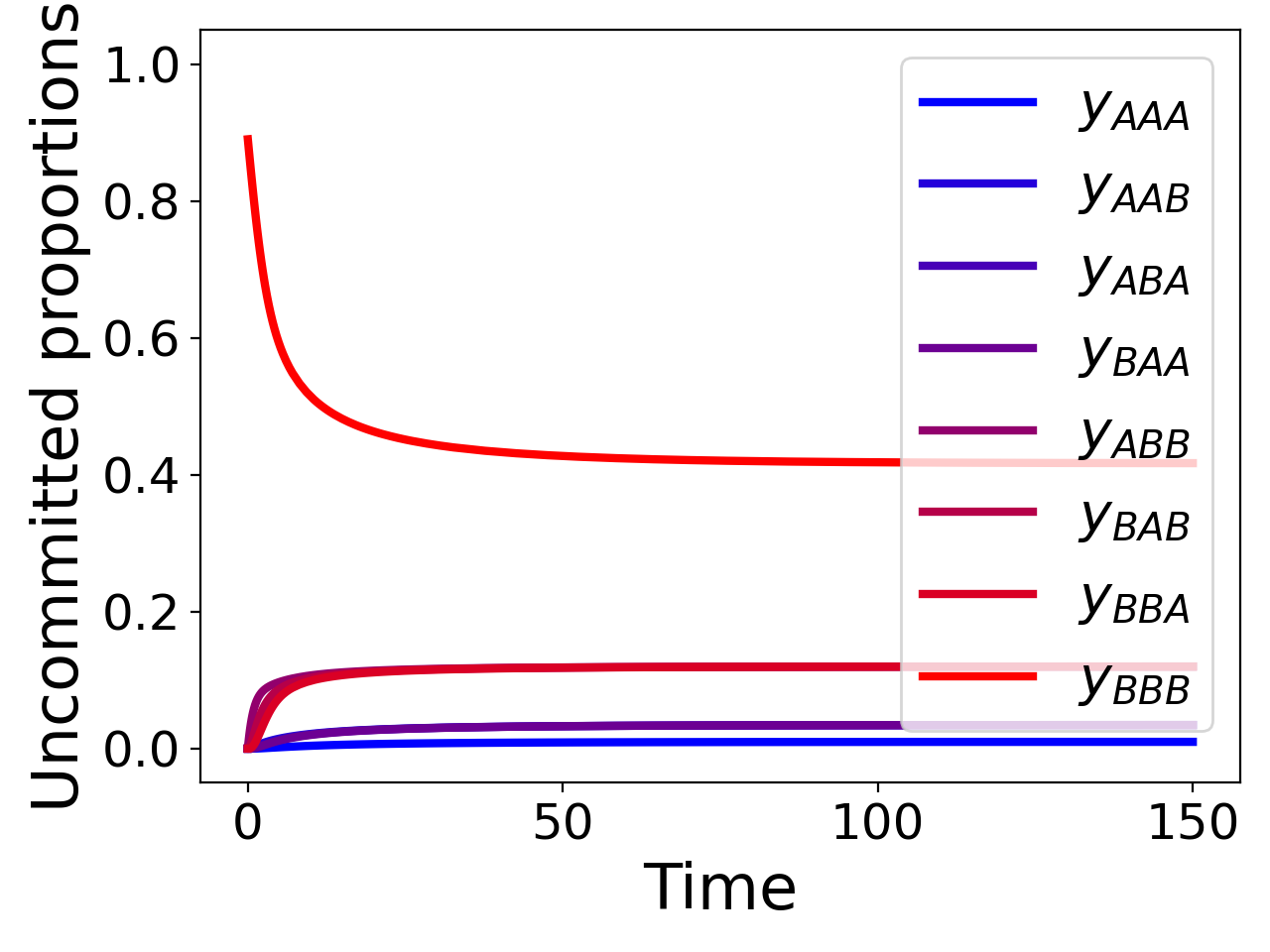} \includegraphics[width=0.48\linewidth]{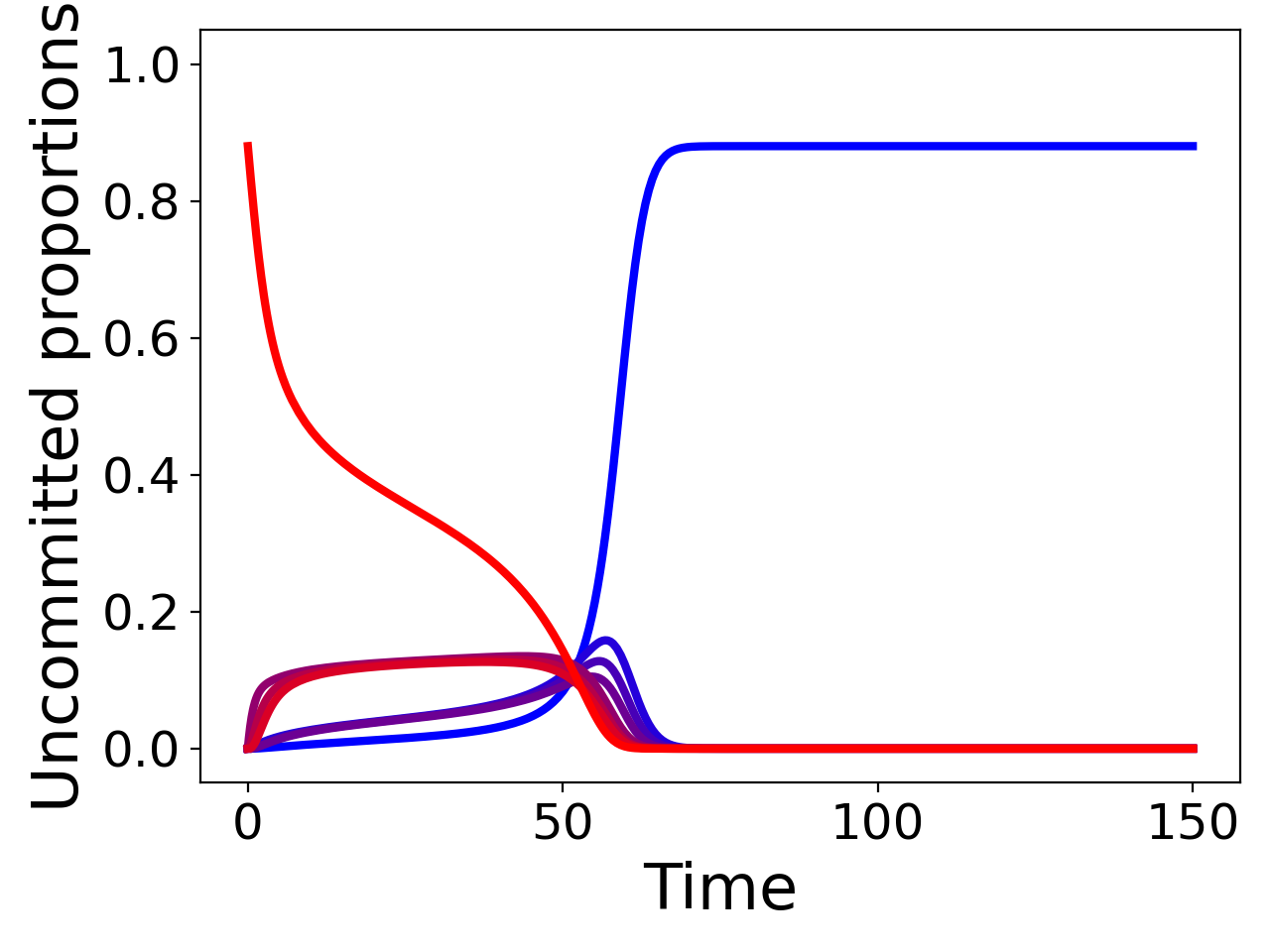}}
    
	\caption{Sample time series for the ODE models with a non-trivial tipping point. Note that the legends are the same within each row.}
\end{figure}

\subsection{Opinion Response Functions} \label{Sect:OpResp_Results}
%-------------------------------------------------------------------------------------
In this section we establish several properties of the opinion response functions $\Phi$ and $\Psi_{\CM}$ defined in \eqref{eq:Phi} and \eqref{eq:Psi} respectively. We begin by establishing a symmetry property of $\Phi$.

\begin{lemma}\label{lem:Phisym}
For all $r\in [0,1]$, $\Phi(1-r)=1-\Phi(r)$.
\end{lemma}
\begin{proof}
We have
\begin{align*}
    \Phi(1-r) &= \P(N(X(1-r))>M/2)+\frac{1}{2}\P(N(X(1-r))= M/2) \\ 
              &= \P(N(X(r))<M/2) + \frac{1}{2}\P(N(X(r))= M/2). \\
    \shortintertext{Then, by taking the complement, we obtain}
    \Phi(1-r) &= 1 - \left(\P(N(X(r))>M/2) + \frac{1}{2}\P(N(X(r))= M/2) \right)\\
              &= 1 - \Phi(r).
\end{align*}
\end{proof}

In Appendix \ref{append:sig_proof} we show that for $M\in\{1,2\}$, $\Phi(r)=r$. For $M\ge 3$, $\Phi$ is a sigmoid function in the following sense.

\begin{theorem}\label{thm:Phisig}
    Let $\Phi$ be given by \eqref{eq:Phi}. For $M\ge 3$, $\Phi(0)=0,\Phi(1/2)=1/2$ and $\Phi(1)=1$, $\Phi$ is increasing on $(0,1)$, convex on $[0,1/2)$ and concave on $(1/2,1]$.
\end{theorem}

\begin{proof}
Let $X(r)$ be as in \eqref{eq:Phi2}. The three fixed points are easily shown: if $r=0$ then $X(r)=0$ with probability 1 so from \eqref{eq:Phi2}, $\Phi(0)=0$. Lemma \ref{lem:Phisym} implies $\Phi(1)=1-\Phi(0)=1$, and plugging in $r=1/2$, $\Phi(1/2)=1-\Phi(1/2)$ so $\Phi(1/2)=1/2$.

To prove convexity/concavity, by Theorem \ref{thm:dbl} below it suffices to consider odd $M$. In this case, $\Phi(r)=\P(N(X(r))>M/2)$. We'll proceed by computing $\Phi'(r)$ but we'd like to do it without a ton of algebra. Represent $X(r)$ as follows: fix i.i.d.~uniform random variables $(u_i)_{i=1}^m$ and say that $X_i(r)=A$ if $u_i \le r$. This is a coupling of the random variables $(X(r)\colon r\in[0,1])$ with the property that if $X_i(r)=A$ and $s>r$ then $X_i(s)=A$. In particular, letting $C(r,h)=\{N(X(r))\le M/2<N(X(r+h))\}$,
\[\Phi(r+h)-\Phi(r) = \P(C(r,h)).\]
If $X_i(r)=B$ and $X_i(r+h)=A$ then $r<u_i \leq r+h$. If $N(X(r))\le M/2<N(X(r+h))$ then up to an event of probability $O(h^2)$ for small $h$, exactly one entry changed between $r$ and $r+h$. In other words, if $\Delta$ is the symmetric difference then
\[\P(C(r,h)\Delta E) = O(h^2)\]
where $E$ is the event that there is a set $I\subset \{1,\dots,M\}$ of size $(M-1)/2$ and an element $j\notin I$ such that $u_i \leq r$ for all $i\in I$, $r<u_j\leq r+h$ and $u_i>r+h$ for all $i\notin I\cup \{j\}$. There are $\binom{M}{(M-1)/2}$ choices for $I$, $(M+1)/2$ choices for $j$, and the probability that the entries are in the required intervals is $hr^{(M-1)/2}(1-(r+h))^{(M-1)/2}$. So,
    \[\P(E)=\binom{M}{(M-1)/2} \frac{M+1}{2} \,hr^{(M-1)/2}(1-(r+h))^{(M-1)/2}.\]
    Notice that
    \[\frac{1}{h}(\Phi(r+h)-\Phi(r)) = \P(E) + O(h).\]
    So, dividing by $h$ and letting $h\to 0$,
    $$\Phi'(r) = \binom{M}{(M-1)/2} \frac{M+1}{2} (r(1-r))^{(M-1)/2}.$$
    In particular, $\Phi'$ is positive, and thus $\Phi$ is increasing, on $(0,1)$, and $\Phi'$ is increasing on $[0,1/2)$, so $\Phi$ is convex on $[0,1/2)$ and, applying Lemma \ref{lem:Phisym}, $\Phi$ is concave on $(1/2,1]$.
\end{proof}

When $\mathcal{C_M}=0$, $\Psi_\mathcal{C_M}(r)=\Phi(r)$ so for $M\ge 3$, $\Psi_{\CM}$ has three fixed points: $0,1/2$ and $1$. When $\mathcal{C_M}=1$, $\Psi_\mathcal{C_M}(r)=1$ for all $r$ and the only fixed point is $1$. In other words, a bifurcation must occur at some $\mathcal{C_M}^*\in(0,1)$. In Appendix \ref{append:bifur_proof}, we show that if $\Phi$ is a sigmoid function then there is a unique $\mathcal{C_M}^*$ at which a saddle-node bifurcation occurs.

As shown in Theorem \ref{thm:ODEeq} there is a one-to-one correspondence between the equilibria of the ODE models and the fixed points of the opinion response functions. Hence, we can use the fixed point equation $\Psi_\mathcal{C_M}(r)=r$ to find the critical value $\mathcal{C_M}^*$ for various $M$ (Figure \ref{fig:OpResp_Tipping}) and plot bifurcation diagrams (Figure \ref{fig:OpResp_Bifur}). We note that the tipping points found using opinion response functions are analytically computed and closely match the estimates found using the ABM. Additionally, using these figures, we can guess the behaviour for all memory bank lengths. In particular, $\mathcal{C_M}^*$ appears to increase with $M$ towards some limit as $M\to\infty$. If the limit exists then it must be $\le 1/2$; see the discussion below the proof of the upcoming Theorem \ref{thm:dbl}.

\begin{figure}[tbph]
    \centering
    \includegraphics[width=0.6\linewidth]{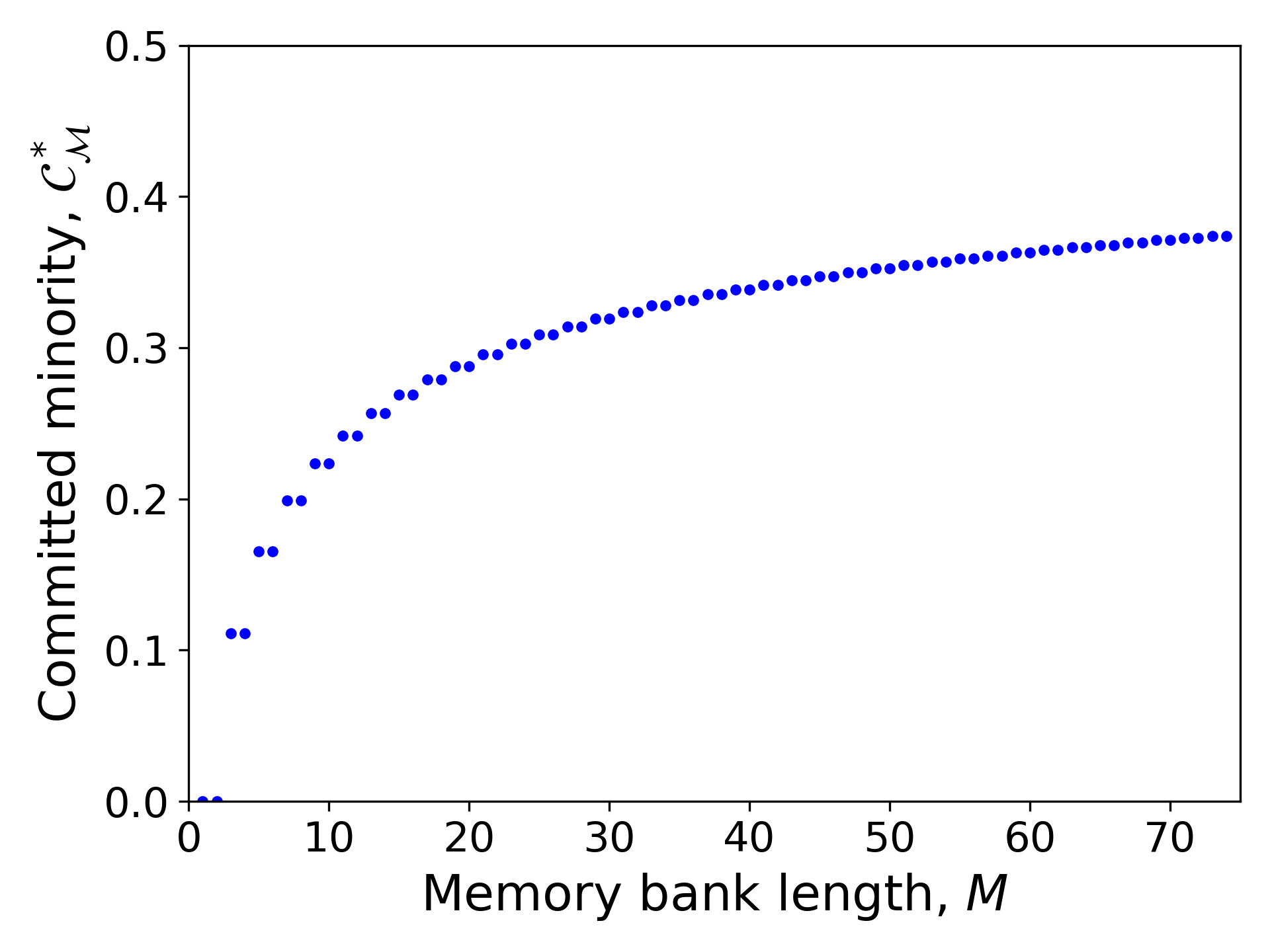}
    \caption{A plot of $\mathcal{C_M}^*$, the minimal committed minority size required to overturn a social convention, versus $M$, the memory bank length, from the opinion response functions.}
    \label{fig:OpResp_Tipping}
\end{figure}

Figures \ref{fig:ABM_TippingPoints} and \ref{fig:OpResp_Tipping} both suggest that the values of $\mathcal{C_M}^*$ are grouped in pairs; for example, $\CM^*$ appears to be $0$ for both $M=1$ and $M=2$. This pattern suggests that increasing the memory bank length from $M-1$ to $M$ for even $M$ does not change the fixed points of $\Psi_{\CM}$. We prove a stronger and somewhat curious result: that the ORFs for $M-1$ and $M$ are identical.

\begin{theorem}\label{thm:dbl}
Denote the function defined in \eqref{eq:Phi} by $\Phi_M$, to emphasize the dependence on the integer $M\ge 1$. If $M$ is even then $\Phi_M=\Phi_{M-1}$.
\end{theorem}
\begin{proof}
    Let $N(X(r))$ have distribution Binomial$(M,r)$, then by definition 
    $$\Phi_M(r) = \P(N(X(r))>M/2) + \frac{1}{2}\P(N(X(r))=M/2).$$
    Let $Y_1, ..., Y_M$ be independent and identically distributed Bernoulli$(r)$ so that $Z_1:=Y_1+...+Y_M$ has distribution Binomial$(M,r)$, i.e., $\Phi_M(r) = \P(Z_1>M/2) + \frac{1}{2}\P(Z_1=M/2)$. If we randomly remove one of the $Y_i$ values from the sum that defines $Z_1$, then we obtain a random variable $Z_2$ with distribution Binomial$(M-1,r)$. Formally, let $I$ be uniformly distributed on $\{1,...,M\}$ independent of $(Y_i)^M_{i=1}$, then $Z_2:=\sum_{i\neq I}Y_i$ has distribution Binomial$(M-1,r)$, i.e., $\Phi_{M-1}(r) = \P(Z_2>(M-1)/2)$. We want to show that
    \[\P(Z_2>(M-1)/2) = \P(Z_1>M/2)+\frac{1}{2}\P(Z_1=M/2).\]
    To do so it suffices to show that
    \begin{enumerate}[noitemsep]
        \item $\P(Z_2>(M-1)/2 \mid Z_1>M/2) = 1$,
        \item $\P(Z_2>(M-1)/2 \mid Z_1<M/2)=0$, and
        \item $\P(Z_2>(M-1)/2 \mid Z_1=M/2)=1/2$.
    \end{enumerate}
    However,
    \begin{enumerate}
        \item if $Z_1>M/2$, then since $M$ is even and $Z_1$ is an integer $Z_1\ge M/2+1$, so $Z_1>M/2+1/2$ and $Z_1-1>(M-1)/2$. Moreover, $Z_2\geq Z_1-1$, so $Z_2>(M-1)/2$.
        \item Similarly, if $Z_1<M/2$, then $Z_1\le M/2-1$, so $Z_2\le Z_1 \le (M-1)/2$.
        \item If $Z_1=M/2$, then $\P(Y_I=1)=1/2$. So, with equal probability we remove either a $0$ or a $1$ from $Z_1$ to obtain $Z_2$. Since $Z_1>(M-1)/2>Z_1-1$, $\P(Z_2>(M-1)/2 \mid Z_1=M/2)=1/2$.
    \end{enumerate}
\end{proof}

In order to determine $\lim_{M\to\infty} \mathcal{C_M}^*$, we must study $\Psi_{\CM}$ for large $M$. If $X(r)$ has distribution Binomial$(M,r)$ and $r$ is fixed then the weak law of large numbers implies that for any $\epsilon>0$, $\P(|X(r) - rM| > \epsilon M)\to 0$ as $M\to\infty$. From the definition of $\Phi$ it follows that
\begin{equation}
   \lim_{M\to\infty} \Phi(r)=\begin{cases} 0 & \text{if} \ r<1/2, \\
    1 & \text{if} \ r>1/2.
              \end{cases} 
\end{equation}
In turn,
\begin{equation}
   \lim_{M\to\infty} \Psi_{\mathcal{C_M}}(r)=\begin{cases} \mathcal{C_M} & \text{if} \ r<1/2, \\
                            1 & \text{if} \ r>1/2.
              \end{cases} 
\end{equation}
Since for finite $M$, $\Psi_{\CM}$ is continuous, it follows that $\mathcal{C_M}^* \to 1/2$ as $M\to\infty$. If we view the large-$M$ limit of $\Psi_{\CM}$ as a function $\Psi^\infty_{\CM}$, then for $\CM<1/2$, $\Psi^\infty_{\CM}$ has fixed points at $\CM$ and $1$, while for $\CM>1/2$ it has only the ``consensus on opinion $A$'' fixed point $r=1$. The behaviour of the equilibria follow a trend where, as $M$ increases, $\mathcal{C_M}^*$ increases and the shape of the bifurcation diagrams becomes more triangular (Figure \ref{fig:OpResp_Bifur}). We note that when $M=1$, we have $\mathcal{C_M}^*=0$ and when $M\to\infty$, we have $\mathcal{C_M}^*=0.5$. These values serve as lower and upper bounds, respectively, for the committed minority size required to cause the overturning of a social convention in the real world. 

\begin{figure}[tbph]
    \centering
    \includegraphics[width=0.6\linewidth]{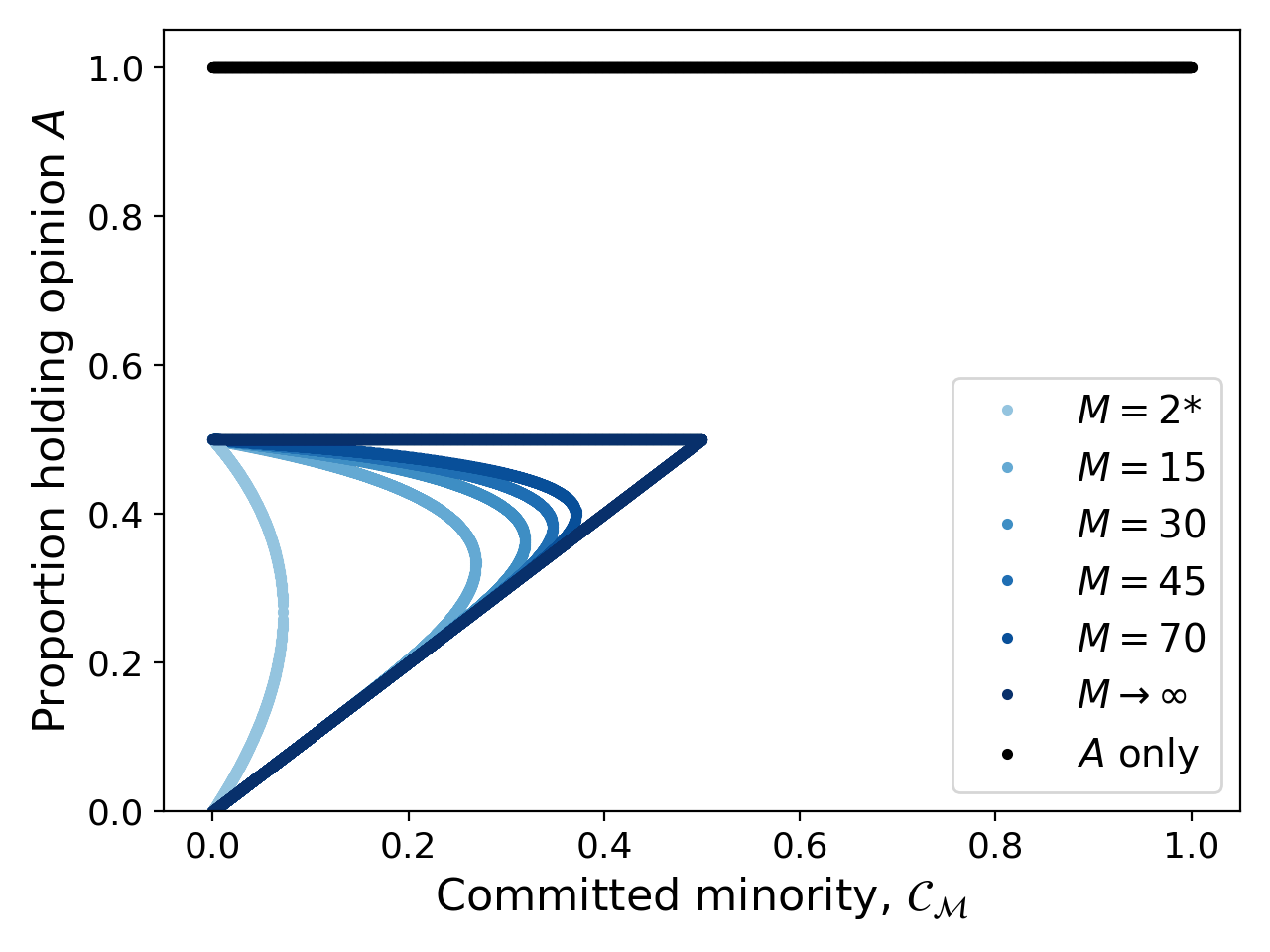}
    \caption{Bifurcation diagrams for various values of $M$. Note that all diagrams have the consensus on $A$ opinion equilibrium ($A$ only) for all $\mathcal{C_M}$ and $M=2^*$ is the $M=2$ with combined undecided individuals case.}
    \label{fig:OpResp_Bifur}
\end{figure}

\section{Discussion} \label{Sect:Discussion}
%-------------------------------------------------------------------------------------
%TIPPING POINTS & MODEL ASSUMPTIONS ------------------------------------------------------------------
Studying equilibria and, if applicable, bifurcations of mathematical models can provide a framework for the expected deterministic model behaviour. \citet{Diekmann2003} introduce a method to decouple model mechanisms and compute the equilibria of a model. The authors apply their framework to a predator-prey model where they decouple the populations from their respective environments. In particular, the prey are part of the predator's environment and vice versa. Making this distinction between population and environment allows the authors to more easily study equilibria and bifurcation structure. In this work, we leverage this idea to study the deterministic behaviour of our opinion dynamics model when the ABM or ODE forms are unwieldy.

We note that the ODE $M=2$ models (with and without combined undecided compartments) have fairly similar definitions, however, the resulting dynamics are very different. In particular, the $M=2$ model has a trivial tipping point, while the $M=2$ with combined undecided compartments has a non-trivial tipping point. Opinion response functions show that this different model behaviour is a result of different functional forms for the speaking rate of opinion $A$. When the undecided compartments are combined we obtain $\Phi(r)=(r^2+r)/2(r^2-r+1)$, which is a sigmoid function, and when these groups are not combined we obtain $\Phi(r)=r$, see Appendix \ref{append:sig_proof}. A possible reason for this difference in functional form is the loss of movement between the $X_{AB}$ and $X_{BA}$ populations (i.e., the interactions where $X_{AB}$ switches to $X_{BA}$ and vice versa). In the $M=2$ model, these switches between $X_{AB}$ and $X_{BA}$ cancel each other out and $\Phi(r)=r$ for all $r$. However, when the undecided compartments are combined, $\Phi(r)<r$ when $r<0.5$. This decreased speaking rate slows down the shift towards the consensus on $A$ state and the result is a non-trivial tipping point.

Our choice to have undecided individuals remain undecided and ``flip a coin'' when they speak is commonly used in modelling \citep{Xie2011, Centola2018, Balenzuela2015}. There are a few additional ways we could have chosen to determine the behaviour of undecided individuals. Following recency effects, an individual might weigh recent memories as more important \citep{Cromwell1950}. On the other hand, an individual might weigh older memories as more important (i.e., primacy effect) if reinforcement learning or the anchoring effect plays a larger role \citep{Lund1925, Knower1936}. Lastly, an individual may choose to not speak when undecided or a modeller may choose not to allow undecided individuals to exist \citep{Baumgaertner2016}. Since we found different model behaviour for the $M=2$ models when we combined the undecided compartments, we hypothesize that the choice of undecided speaking behaviour can have a strong effect on model dynamics and warrants further study. 

%CENTERING AND DETERRENTS TO INTERACTION ---------------------------------------------------------
In our $M\in\{2,3\}$ ODE models, we note that the moderate opinion groups (i.e., those not holding all $A$ or all $B$ in their memory banks) initially increase after the tipping point is crossed and then decrease as the consensus on $A$ behaviour takes over. This initial behaviour is called ``centering'' and has been found in other opinion dynamics models \citep{Baumgaertner2016, Baumgaertner2018a, Iacopini2022}. In \citet{Baumgaertner2016, Baumgaertner2018a}, the authors assume that an interaction between two agents holding the same opinion will result in amplification of that opinion. When the probability of this amplification increases, they find that the system reaches consensus faster and there is a smaller degree of centering. This behaviour is similar to our ODE models when we increase $\mathcal{C_M}$. In \citep{Iacopini2022}, however, the intermediate populations can persist and do not always decay. The difference may lie in the authors' use of a parameter to describe the tendency of individuals to update their opinion or not following an interaction. In particular, when the individuals are less likely to update their opinions, the centering behaviour remains \citep{Iacopini2022}.

In this manuscript, we use ``listening'' and ``hearing'' interchangeably since the listeners in our models always listen to and take in the opinions that they hear. We could have instead defined an opinion difference threshold above which individuals do not update their memory banks \citep{Deffuant2000, Weisbuch2002}. Previous work has also studied the dynamics when this threshold varies between individuals in the population \citep{Granovetter1978}. In these bounded confidence models, high thresholds can result in consensus and low thresholds can result in multiple polarized groups. An individual may be less likely to listen if they are stubborn \citep{Biswas2009}, hold their opinion strongly \citep{Baumgaertner2016}, or hold vested interests \citep{Crano1995, DeDominicis2021, Gussmann2021}. The authors find that when individuals are allowed to increase their opinion strength or become stubborn, the result is polarization of opinions \citep{Biswas2009, Baumgaertner2016}. An extension to our social model is formulating a rule where individuals become stubborn and do not listen when their memory banks are too different from each other. In order to study the model dynamics, we must consider a population with varying opinion thresholds or redefine our initial conditions. Otherwise, the committed minority (who only have memories of type $A$) won't be able to convince anyone in the uncommitted population (who initially have only memories of type $B$) to listen and there will be no change in the distribution of opinions. We hypothesize that, similar to other models with stubbornness or opinion thresholds, the resulting model behaviour will be polarization.

Another deterrent to interaction and the sharing of opinions is the inability to interact. In our models, we assume the system is well-mixed and that every individual is able to interact with any other individual. Incomplete graphs can be used to model opinion dynamics systems when it is not guaranteed that every individual can interact with every other individual \citep{Mariano2020}. These incomplete graphs can represent echo chambers and the authors find model dynamics result in polarization. More complex network models can use asymmetric graphs \citep{Talebi2024, Su2021}. In such cases, there exist pairs of individuals where one individual can speak their opinion to another individual, but the other individual is less able to or cannot speak back due to, e.g., lack of influence \citep{Talebi2024}, disability \citep{Liu2018}, or social status \citep{Su2021, HasaniMavriqi2016}. We expect that using an incomplete graph will increase the size of committed minority required to overturn a social convention.

Previous work has found that variables other than committed minority size can also influence opinion dynamics and the occurrence of tipping events \citep{Andersson2021, Iacopini2022, Nowak2004}. \citet{Andersson2021} study a bounded confidence model with an external field that represents factors such as government policy, or mass media. The authors find that increasing the external field causes a tipping event where below the tipping point, agents are more likely to defect and above the tipping point, agents are more likely to cooperate. \citep{Iacopini2022} obtain a similar result and find that consensus becomes possible above a certain listening threshold. Below this threshold, the individuals are less likely to listen or update their opinions following an interaction, resulting in polarization and the majority of the population being undecided \citep{Iacopini2022}. In \citep{Nowak2004}, the population size can determine whether a committed individual is capable of upsetting the defection social convention and achieving consensus on cooperation. In particular, the authors note that it is difficult to reach the minimum critical mass in large populations while small populations may not have enough individuals for cooperation to win in a cost-benefit analysis. 

%COMMITTED MINORITY & OTHER TIPPING POINTS ----------------------------------------------------------
A key result from our work is that there is indeed a minimum committed minority size below which tipping of a social convention cannot occur when the memory bank length is at least three. Similar $\mathcal{C_M^*}$ are observed in models based on the Voter model \citep{Holley1975}, the naming game \citep{Baronchelli2006, Xie2011}, and a two state model with group interactions \citep{Galam2007, Galam2010, Galam2020, Iacopini2022}. Our results also compare well with those from \citet{Centola2018} who use the same interaction rules as our models. While our result appears different from their non-trivial tipping point when $M=2$, this apparent discrepancy is an artifact of their choice of a larger error bound. We note that the model dynamics from each of these models are similar, even though individuals in each model type interact following different rules. This similarity in results suggests that how the individuals interact does not have a large impact on model dynamics and the size of the committed minority is the main driver of behaviour in these models.

The committed minority, however, need not be one sided. Various papers consider multiple committed minorities, one for each opinion \citep{Mobilia2007, Biswas2009, Galam2007, Galam2010, Galam2020}. They find that committed minorities of equal size and opposite opinion prevent consensus \citep{Mobilia2007} and the unstable steady state is at exactly $r_A=r_B=0.5$ \citep{Galam2007}. That is, if opinion $A$ is initially the social convention, then it remains the majority opinion and vice versa \citep{Galam2007}. These results suggest that the ``pulling up'' of the opinion response functions from the opinion $A$ committed minority is exactly balanced by a ``pulling down'' from the opinion $B$ committed minority when the minorities are equal in size. When the minorities are not equal in size, the result is consensus on the opinion of the larger committed minority \citep{Biswas2009, Mobilia2007, Galam2007, Galam2010, Galam2020}. Expanding our opinion response functions to include a committed minority holding opinions $A$ and $B$ would allow us to gain further insights on human behaviour when there are vested interests on both sides of an argument. In particular, we could study the model behaviour for varying committed minority sizes and determine if the bifurcation from coexistence to consensus becomes reversible if there are two committed minorities.

\section{Acknowledgement}
We thank Rebecca C. Tyson and Bert Baumgaertner for their support. This work was supported by Natural Sciences and Engineering Research Council of Canada [grant numbers RGPIN-2024-04653, CGS-M]. The funding sources were not involved in the conduction of research or preparation of this article. 

\section{Code Availability}
The associated code is available at \url{https://github.com/sarahwyse/OpinionResponseFunctions.git} and the simulations in this manuscript are run in Python 3.9.12 (Spyder 5.4.3).

\bibliographystyle{plainnat}
\bibliography{Sources.bib}

\begin{thebibliography}{45}
\providecommand{\natexlab}[1]{#1}
\providecommand{\url}[1]{\texttt{#1}}
\expandafter\ifx\csname urlstyle\endcsname\relax
  \providecommand{\doi}[1]{doi: #1}\else
  \providecommand{\doi}{doi: \begingroup \urlstyle{rm}\Url}\fi

\bibitem[Andersson et~al.(2021)Andersson, Bratsberg, Ringsmuth, and de~Wijn]{Andersson2021}
David Andersson, Sigrid Bratsberg, Andrew~K. Ringsmuth, and Astrid~S. de~Wijn.
\newblock Dynamics of collective action to conserve a large common-pool resource.
\newblock \emph{Scientific Reports}, 11\penalty0 (1), 2021.
\newblock ISSN 2045-2322.
\newblock \doi{10.1038/s41598-021-87109-x}.
\newblock URL \url{http://dx.doi.org/10.1038/s41598-021-87109-x}.

\bibitem[Balenzuela et~al.(2015)Balenzuela, Pinasco, and Semeshenko]{Balenzuela2015}
Pablo Balenzuela, Juan~Pablo Pinasco, and Viktoriya Semeshenko.
\newblock The undecided have the key: Interaction-driven opinion dynamics in a three state model.
\newblock \emph{PLOS ONE}, 10\penalty0 (10):\penalty0 e0139572, October 2015.
\newblock ISSN 1932-6203.
\newblock \doi{10.1371/journal.pone.0139572}.
\newblock URL \url{http://dx.doi.org/10.1371/journal.pone.0139572}.

\bibitem[Baronchelli et~al.(2006)Baronchelli, Felici, Loreto, Caglioti, and Steels]{Baronchelli2006}
Andrea Baronchelli, Maddalena Felici, Vittorio Loreto, Emanuele Caglioti, and Luc Steels.
\newblock Sharp transition towards shared vocabularies in multi-agent systems.
\newblock \emph{Journal of Statistical Mechanics: Theory and Experiment}, 2006\penalty0 (06):\penalty0 P06014--P06014, 2006.
\newblock \doi{10.1088/1742-5468/2006/06/p06014}.
\newblock URL \url{http://dx.doi.org/10.1088/1742-5468/2006/06/p06014}.

\bibitem[Baumgaertner et~al.(2016)Baumgaertner, Tyson, and Krone]{Baumgaertner2016}
Bert~O. Baumgaertner, Rebecca~C. Tyson, and Stephen~M. Krone.
\newblock Opinion strength influences the spatial dynamics of opinion formation.
\newblock \emph{The Journal of Mathematical Sociology}, 40\penalty0 (4):\penalty0 207--218, 2016.
\newblock \doi{10.1080/0022250x.2016.1205049}.
\newblock URL \url{http://dx.doi.org/10.1080/0022250x.2016.1205049}.

\bibitem[Baumgaertner et~al.(2018)Baumgaertner, Fetros, Krone, and Tyson]{Baumgaertner2018a}
Bert~O. Baumgaertner, Peter~A. Fetros, Stephen~M. Krone, and Rebecca~C. Tyson.
\newblock Spatial opinion dynamics and the effects of two types of mixing.
\newblock \emph{Physical Review E}, 98\penalty0 (2), 2018.
\newblock ISSN 2470-0053.
\newblock \doi{10.1103/physreve.98.022310}.
\newblock URL \url{http://dx.doi.org/10.1103/physreve.98.022310}.

\bibitem[Biswas and Sen(2009)]{Biswas2009}
Soham Biswas and Parongama Sen.
\newblock Model of binary opinion dynamics: Coarsening and effect of disorder.
\newblock \emph{Physical Review E}, 80\penalty0 (2), 2009.
\newblock ISSN 1550-2376.
\newblock \doi{10.1103/physreve.80.027101}.
\newblock URL \url{http://dx.doi.org/10.1103/physreve.80.027101}.

\bibitem[Bohns et~al.(2013)Bohns, Roghanizad, and Xu]{Bohns2013}
Vanessa~K. Bohns, M.~Mahdi Roghanizad, and Amy~Z. Xu.
\newblock Underestimating our influence over others' unethical behavior and decisions.
\newblock \emph{Personality and Social Psychology Bulletin}, 40\penalty0 (3):\penalty0 348--362, 2013.
\newblock \doi{10.1177/0146167213511825}.
\newblock URL \url{http://dx.doi.org/10.1177/0146167213511825}.

\bibitem[Cao et~al.(2015)Cao, Wang, and Qiao]{Cao2015}
J~Cao, H~Wang, and F~Qiao.
\newblock \emph{Extending the Deffuant model by incorporating the influence factor}, page 11–16.
\newblock CRC Press, February 2015.
\newblock ISBN 9781315752297.
\newblock \doi{10.1201/b18049-4}.
\newblock URL \url{http://dx.doi.org/10.1201/b18049-4}.

\bibitem[Centola et~al.(2018)Centola, Becker, Brackbill, and Baronchelli]{Centola2018}
Damon Centola, Joshua Becker, Devon Brackbill, and Andrea Baronchelli.
\newblock Experimental evidence for tipping points in social convention.
\newblock \emph{Science}, 360\penalty0 (6393):\penalty0 1116--1119, 2018.
\newblock \doi{10.1126/science.aas8827}.
\newblock URL \url{http://dx.doi.org/10.1126/science.aas8827}.

\bibitem[Constantino et~al.(2022)Constantino, Sparkman, Kraft-Todd, Bicchieri, Centola, Shell-Duncan, Vogt, and Weber]{Constantino2022}
Sara~M. Constantino, Gregg Sparkman, Gordon~T. Kraft-Todd, Cristina Bicchieri, Damon Centola, Bettina Shell-Duncan, Sonja Vogt, and Elke~U. Weber.
\newblock Scaling up change: A critical review and practical guide to harnessing social norms for climate action.
\newblock \emph{Psychological Science in the Public Interest}, 23\penalty0 (2):\penalty0 50--97, 2022.
\newblock \doi{10.1177/15291006221105279}.
\newblock URL \url{http://dx.doi.org/10.1177/15291006221105279}.

\bibitem[Crano and Prislin(1995)]{Crano1995}
William~D. Crano and Radmila Prislin.
\newblock Components of vested interest and attitude-behavior consistency.
\newblock \emph{Basic and Applied Social Psychology}, 17\penalty0 (1–2):\penalty0 1–21, August 1995.
\newblock ISSN 1532-4834.
\newblock \doi{10.1080/01973533.1995.9646129}.
\newblock URL \url{http://dx.doi.org/10.1080/01973533.1995.9646129}.

\bibitem[Cromwell(1950)]{Cromwell1950}
Harvey Cromwell.
\newblock The relative effect on audience attitude of the first versus the second argumentative speech of a series.
\newblock \emph{Speech Monographs}, 17\penalty0 (2):\penalty0 105–122, June 1950.
\newblock ISSN 0038-7169.
\newblock \doi{10.1080/03637755009375004}.
\newblock URL \url{http://dx.doi.org/10.1080/03637755009375004}.

\bibitem[De~Dominicis et~al.(2021)De~Dominicis, Ganucci~Cancellieri, Crano, Stancu, and Bonaiuto]{DeDominicis2021}
Stefano De~Dominicis, Uberta Ganucci~Cancellieri, William~D. Crano, Alexandra Stancu, and Marino Bonaiuto.
\newblock Experiencing, caring, coping: Vested interest mediates the effect of past experience on coping behaviors in environmental risk contexts.
\newblock \emph{Journal of Applied Social Psychology}, 51\penalty0 (3):\penalty0 286–304, January 2021.
\newblock ISSN 1559-1816.
\newblock \doi{10.1111/jasp.12735}.
\newblock URL \url{http://dx.doi.org/10.1111/jasp.12735}.

\bibitem[Deffuant et~al.(2000)Deffuant, Neau, Amblard, and Weisbuch]{Deffuant2000}
Guillaume Deffuant, David Neau, Frederic Amblard, and G{\'{e}}rard Weisbuch.
\newblock Mixing beliefs among interacting agents.
\newblock \emph{Advances in Complex Systems}, 03\penalty0 (01n04):\penalty0 87--98, 2000.
\newblock \doi{10.1142/s0219525900000078}.
\newblock URL \url{http://dx.doi.org/10.1142/s0219525900000078}.

\bibitem[Diekmann et~al.(2003)Diekmann, Gyllenberg, and Metz]{Diekmann2003}
O.~Diekmann, M.~Gyllenberg, and J.A.J. Metz.
\newblock Steady-state analysis of structured population models.
\newblock \emph{Theoretical Population Biology}, 63\penalty0 (4):\penalty0 309--338, 2003.
\newblock \doi{10.1016/s0040-5809(02)00058-8}.
\newblock URL \url{http://dx.doi.org/10.1016/s0040-5809(02)00058-8}.

\bibitem[Flynn and Lake(2008)]{Flynn2008}
Francis~J. Flynn and Vanessa K.~B. Lake.
\newblock If you need help, just ask: Underestimating compliance with direct requests for help.
\newblock \emph{Journal of Personality and Social Psychology}, 95\penalty0 (1):\penalty0 128--143, 2008.
\newblock \doi{10.1037/0022-3514.95.1.128}.
\newblock URL \url{http://dx.doi.org/10.1037/0022-3514.95.1.128}.

\bibitem[Foxall(2018)]{Foxall2018}
Eric Foxall.
\newblock The naming game on the complete graph.
\newblock \emph{Electronic Journal of Probability}, 23, 2018.
\newblock ISSN 1083-6489.
\newblock \doi{10.1214/18-ejp250}.
\newblock URL \url{http://dx.doi.org/10.1214/18-ejp250}.

\bibitem[Galam(2010)]{Galam2010}
Serge Galam.
\newblock Public debates driven by incomplete scientific data: The cases of evolution theory, global warming and {H1N1} pandemic influenza.
\newblock \emph{Physica A: Statistical Mechanics and its Applications}, 389\penalty0 (17):\penalty0 3619–3631, September 2010.
\newblock ISSN 0378-4371.
\newblock \doi{10.1016/j.physa.2010.04.039}.
\newblock URL \url{http://dx.doi.org/10.1016/j.physa.2010.04.039}.

\bibitem[Galam and Cheon(2020)]{Galam2020}
Serge Galam and Taksu Cheon.
\newblock Tipping points in opinion dynamics: A universal formula in five dimensions.
\newblock \emph{Frontiers in Physics}, 8, November 2020.
\newblock ISSN 2296-424X.
\newblock \doi{10.3389/fphy.2020.566580}.
\newblock URL \url{http://dx.doi.org/10.3389/fphy.2020.566580}.

\bibitem[Galam and Jacobs(2007)]{Galam2007}
Serge Galam and Frans Jacobs.
\newblock The role of inflexible minorities in the breaking of democratic opinion dynamics.
\newblock \emph{Physica A: Statistical Mechanics and its Applications}, 381:\penalty0 366--376, 2007.
\newblock \doi{10.1016/j.physa.2007.03.034}.
\newblock URL \url{http://dx.doi.org/10.1016/j.physa.2007.03.034}.

\bibitem[Granovetter(1978)]{Granovetter1978}
Mark Granovetter.
\newblock Threshold models of collective behavior.
\newblock \emph{The American journal of sociology}, 83\penalty0 (6):\penalty0 1420--1443, 1978.
\newblock URL \url{https://doi.org/10.1086/226707}.

\bibitem[Gussmann and Hinkel(2021)]{Gussmann2021}
Geronimo Gussmann and Jochen Hinkel.
\newblock Vested interests, rather than adaptation considerations, explain varying post-tsunami relocation outcomes in {L}aamu atoll, {M}aldives.
\newblock \emph{One Earth}, 4\penalty0 (10):\penalty0 1468–1476, October 2021.
\newblock ISSN 2590-3322.
\newblock \doi{10.1016/j.oneear.2021.09.004}.
\newblock URL \url{http://dx.doi.org/10.1016/j.oneear.2021.09.004}.

\bibitem[Hasani-Mavriqi et~al.(2016)Hasani-Mavriqi, Geigl, Pujari, Lex, and Helic]{HasaniMavriqi2016}
Ilire Hasani-Mavriqi, Florian Geigl, Subhash~Chandra Pujari, Elisabeth Lex, and Denis Helic.
\newblock The influence of social status and network structure on consensus building in collaboration networks.
\newblock \emph{Social Network Analysis and Mining}, 6\penalty0 (1), September 2016.
\newblock ISSN 1869-5469.
\newblock \doi{10.1007/s13278-016-0389-y}.
\newblock URL \url{http://dx.doi.org/10.1007/s13278-016-0389-y}.

\bibitem[Holley and Liggett(1975)]{Holley1975}
Richard~A. Holley and Thomas~M. Liggett.
\newblock Ergodic theorems for weakly interacting infinite systems and the voter model.
\newblock \emph{The Annals of Probability}, 3\penalty0 (4):\penalty0 643--663, 1975.
\newblock ISSN 00911798.
\newblock URL \url{http://www.jstor.org/stable/2959329}.

\bibitem[Iacopini et~al.(2022)Iacopini, Petri, Baronchelli, and Barrat]{Iacopini2022}
Iacopo Iacopini, Giovanni Petri, Andrea Baronchelli, and Alain Barrat.
\newblock Group interactions modulate critical mass dynamics in social convention.
\newblock \emph{Communications Physics}, 5\penalty0 (1), 2022.
\newblock \doi{10.1038/s42005-022-00845-y}.
\newblock URL \url{http://dx.doi.org/10.1038/s42005-022-00845-y}.

\bibitem[J\k{e}drzejewski and Sznajd-Weron(2018)]{Jedrzejewski2018}
Arkadiusz J\k{e}drzejewski and Katarzyna Sznajd-Weron.
\newblock Impact of memory on opinion dynamics.
\newblock \emph{Physica A: Statistical Mechanics and its Applications}, 505:\penalty0 306–315, 2018.
\newblock ISSN 0378-4371.
\newblock \doi{10.1016/j.physa.2018.03.077}.
\newblock URL \url{http://dx.doi.org/10.1016/j.physa.2018.03.077}.

\bibitem[Knower(1936)]{Knower1936}
F.~H. Knower.
\newblock Experimental studies of changes in attitude. ii. a study of the effect of printed argument on changes in attitude.
\newblock \emph{The Journal of Abnormal and Social Psychology}, 30\penalty0 (4):\penalty0 522–532, January 1936.
\newblock ISSN 0096-851X.
\newblock \doi{10.1037/h0055902}.
\newblock URL \url{http://dx.doi.org/10.1037/h0055902}.

\bibitem[Kurtz(1976)]{Kurtz1976}
Thomas~G. Kurtz.
\newblock \emph{Limit theorems and diffusion approximations for density dependent Markov chains}, page 67–78.
\newblock Springer Berlin Heidelberg, 1976.
\newblock ISBN 9783642007842.
\newblock \doi{10.1007/bfb0120765}.
\newblock URL \url{http://dx.doi.org/10.1007/BFb0120765}.

\bibitem[Liu et~al.(2018)Liu, Xie, Han, Mou, Zhang, and Zhang]{Liu2018}
Shen Liu, Wenlan Xie, Shangfeng Han, Zhongchen Mou, Xiaochu Zhang, and Lin Zhang.
\newblock Social interaction patterns of the disabled people in asymmetric social dilemmas.
\newblock \emph{Frontiers in Psychology}, 9, September 2018.
\newblock ISSN 1664-1078.
\newblock \doi{10.3389/fpsyg.2018.01683}.
\newblock URL \url{http://dx.doi.org/10.3389/fpsyg.2018.01683}.

\bibitem[Lund(1925)]{Lund1925}
F.~H. Lund.
\newblock The psychology of belief.
\newblock \emph{The Journal of Abnormal and Social Psychology}, 20\penalty0 (1):\penalty0 63–195, April 1925.
\newblock ISSN 0096-851X.
\newblock \doi{10.1037/h0076047}.
\newblock URL \url{http://dx.doi.org/10.1037/h0076047}.

\bibitem[Majmudar et~al.(2019)Majmudar, Krone, Baumgaertner, and Tyson]{Majmudar2019}
Jimit~R. Majmudar, Stephen~M. Krone, Bert~O. Baumgaertner, and Rebecca~C. Tyson.
\newblock Voter models and external influence.
\newblock \emph{The Journal of Mathematical Sociology}, 44\penalty0 (1):\penalty0 1–11, 2019.
\newblock ISSN 1545-5874.
\newblock \doi{10.1080/0022250x.2019.1625349}.
\newblock URL \url{http://dx.doi.org/10.1080/0022250X.2019.1625349}.

\bibitem[Mariano et~al.(2020)Mariano, Morarescu, Postoyan, and Zaccarian]{Mariano2020}
S.~Mariano, I.~C. Morarescu, R.~Postoyan, and L.~Zaccarian.
\newblock A hybrid model of opinion dynamics with memory-based connectivity.
\newblock \emph{{IEEE} Control Systems Letters}, 4\penalty0 (3):\penalty0 644--649, 2020.
\newblock \doi{10.1109/lcsys.2020.2989077}.
\newblock URL \url{http://dx.doi.org/10.1109/lcsys.2020.2989077}.

\bibitem[Mobilia et~al.(2007)Mobilia, Petersen, and Redner]{Mobilia2007}
M~Mobilia, A~Petersen, and S~Redner.
\newblock On the role of zealotry in the voter model.
\newblock \emph{Journal of Statistical Mechanics: Theory and Experiment}, 2007\penalty0 (08):\penalty0 P08029–P08029, 2007.
\newblock ISSN 1742-5468.
\newblock \doi{10.1088/1742-5468/2007/08/p08029}.
\newblock URL \url{http://dx.doi.org/10.1088/1742-5468/2007/08/p08029}.

\bibitem[Nowak et~al.(2004)Nowak, Sasaki, Taylor, and Fudenberg]{Nowak2004}
Martin~A. Nowak, Akira Sasaki, Christine Taylor, and Drew Fudenberg.
\newblock Emergence of cooperation and evolutionary stability in finite populations.
\newblock \emph{Nature}, 428\penalty0 (6983):\penalty0 646–650, April 2004.
\newblock ISSN 1476-4687.
\newblock \doi{10.1038/nature02414}.
\newblock URL \url{http://dx.doi.org/10.1038/nature02414}.

\bibitem[Papst et~al.(2022)Papst, O'Keeffe, and Strogatz]{Papst2022}
Irena Papst, Kevin~P. O'Keeffe, and Steven~H. Strogatz.
\newblock Modeling the interplay between seasonal flu outcomes and individual vaccination decisions.
\newblock \emph{Bulletin of Mathematical Biology}, 84\penalty0 (3), 2022.
\newblock \doi{10.1007/s11538-021-00988-z}.
\newblock URL \url{http://dx.doi.org/10.1007/s11538-021-00988-z}.

\bibitem[Smaldino(2017)]{Smaldino2017}
Paul~E. Smaldino.
\newblock \emph{Models Are Stupid, and We Need More of Them}, page 311–331.
\newblock Routledge, May 2017.
\newblock ISBN 9781315173726.
\newblock \doi{10.4324/9781315173726-14}.
\newblock URL \url{http://dx.doi.org/10.4324/9781315173726-14}.

\bibitem[Sood and Redner(2005)]{Sood2005}
V.~Sood and S.~Redner.
\newblock Voter model on heterogeneous graphs.
\newblock \emph{Physical Review Letters}, 94\penalty0 (17), 2005.
\newblock ISSN 1079-7114.
\newblock \doi{10.1103/physrevlett.94.178701}.
\newblock URL \url{http://dx.doi.org/10.1103/PhysRevLett.94.178701}.

\bibitem[Su et~al.(2021)Su, Allen, and Plotkin]{Su2021}
Qi~Su, Benjamin Allen, and Joshua~B. Plotkin.
\newblock Evolution of cooperation with asymmetric social interactions.
\newblock \emph{Proceedings of the National Academy of Sciences}, 119\penalty0 (1), December 2021.
\newblock ISSN 1091-6490.
\newblock \doi{10.1073/pnas.2113468118}.
\newblock URL \url{http://dx.doi.org/10.1073/pnas.2113468118}.

\bibitem[Talebi et~al.(2024)Talebi, Boroujeni, and Razi]{Talebi2024}
Amirreza Talebi, Sayed Pedram~Haeri Boroujeni, and Abolfazl Razi.
\newblock Opinion dynamics in social multiplex networks with mono and bi-directional interactions in the presence of leaders, 2024.
\newblock URL \url{https://arxiv.org/abs/2401.15857}.

\bibitem[Tessone et~al.(2013)Tessone, Sánchez, and Schweitzer]{Tessone2013}
Claudio~J. Tessone, Angel Sánchez, and Frank Schweitzer.
\newblock Diversity-induced resonance in the response to social norms.
\newblock \emph{Physical Review E}, 87\penalty0 (2), 2013.
\newblock ISSN 1550-2376.
\newblock \doi{10.1103/physreve.87.022803}.
\newblock URL \url{http://dx.doi.org/10.1103/physreve.87.022803}.

\bibitem[van~den Bergh et~al.(2019)van~den Bergh, Savin, and Drews]{vandenBergh2019}
Jeroen~C.J.M. van~den Bergh, Ivan Savin, and Stefan Drews.
\newblock Evolution of opinions in the growth-vs-environment debate: Extended replicator dynamics.
\newblock \emph{Futures}, 109:\penalty0 84–100, 2019.
\newblock ISSN 0016-3287.
\newblock \doi{10.1016/j.futures.2019.02.024}.
\newblock URL \url{http://dx.doi.org/10.1016/j.futures.2019.02.024}.

\bibitem[Weisbuch et~al.(2002)Weisbuch, Deffuant, Amblard, and Nadal]{Weisbuch2002}
Gerard Weisbuch, Guillaume Deffuant, Frederic Amblard, and Jean-Pierre Nadal.
\newblock Meet, discuss, and segregate!
\newblock \emph{Complexity}, 7\penalty0 (3):\penalty0 55--63, 2002.
\newblock \doi{10.1002/cplx.10031}.
\newblock URL \url{http://dx.doi.org/10.1002/cplx.10031}.

\bibitem[Wimsatt(2007)]{Wimsatt2007}
W.C. Wimsatt.
\newblock \emph{False Models as Means to Truer Theories}, pages 94--132.
\newblock Harvard University Press, 2007.
\newblock \doi{10.2307/j.ctv1pncnrh.10}.
\newblock URL \url{http://dx.doi.org/10.2307/j.ctv1pncnrh.10}.

\bibitem[Wyse(2024)]{Wyse2024}
Sarah~K Wyse.
\newblock \emph{The Role of Committed Minorities in Climate Change Action: Qualitative Insights from a Social-Climate Model}.
\newblock {MSc} thesis, University of British Columbia - Okanagan, 2024.
\newblock Available at \url{http://hdl.handle.net/2429/88472 }.

\bibitem[Xie et~al.(2011)Xie, Sreenivasan, Korniss, Zhang, Lim, and Szymanski]{Xie2011}
J.~Xie, S.~Sreenivasan, G.~Korniss, W.~Zhang, C.~Lim, and B.~K. Szymanski.
\newblock Social consensus through the influence of committed minorities.
\newblock \emph{Physical Review E}, 84\penalty0 (1), 2011.
\newblock \doi{10.1103/physreve.84.011130}.
\newblock URL \url{http://dx.doi.org/10.1103/physreve.84.011130}.

\end{thebibliography}

\begin{appendices}
\section{ODEs for the modified $M=2$ case} \label{append:ode_eqs}
%--------------------------------------------------------------------------------------------------------

The speaking rate of opinion $A$ is $r_A=y_{AA}+\tfrac{1}{2}y_{U}+\mathcal{C_M}$ and the speaking rate of opinion $B$ is $r_B=y_{BB}+\tfrac{1}{2}y_{U}$. Hence, the full system of equations for the $M=2$ model with a combined undecided group is: 
\begin{align}\label{sup:eq:Model_M=2_simp}
	\frac{dy_{AA}}{dt} &= r_A(y)y_{U} - r_B(y)y_{AA}, \nonumber \\
    \frac{dy_{U}}{dt} &= r_B(y)y_{AA} + r_A(y)y_{BB} - (r_A(y) + r_B(y))y_{U}, \\
    \frac{dy_{BB}}{dt} &= r_B(y)y_{U} - r_A(y)y_{BB}. \nonumber 
\end{align}

\section{Computation of $\Phi$ for $M\le 3$ and \\ the modified $M=2$ case} 
\label{append:sig_proof}
%--------------------------------------------------------------------------------------------------------
\underline{$M=1$:} $$\Phi(r) = \P(N(X(r))=1) = r$$

\underline{$M=2$:} $$\Phi(r) = \P(N(X(r))=2)+\tfrac{1}{2}\P(N(X(r))=1) = r^2+\frac{1}{2}2r(1-r) = r$$

\underline{$M=2$ with combined undecided individuals:} \\ 
For this model we take $\Phi$ to have the same form as in \eqref{eq:Phi} except with $S=\{A,B,U\}$, $s(A)=1$, $s(U)=\tfrac{1}{2}$ and $s(B)=0$ and $\pi_r$ equal to the stationary distribution of the Markov chain on $\{A,U,B\}$ with transitions $A\to U$ and $U\to B$ at rate $1-r$ and $B\to U$, $U\to A$ at rate $r$. Following the same approach as in Theorem \ref{thm:ODEeq} (omitted) it is not hard to prove an analogous result, i.e., that $y$ is an equilibrium of \eqref{sup:eq:Model_M=2_simp} iff $y=(1-\CM)\pi_r$ for some $r$ that satisfies $r=\Psi_{\CM}(r)$, where $\Psi_{\CM}=\CM+(1-\CM)\Phi$.\\

To compute $\Phi$, first note that since the Markov chain has no minimal loops of length greater than $2$, $\pi_r$ satisfies detailed balance, i.e., $(1-r)\pi_r(A)=r\pi_r(U)$ and $(1-r)\pi_r(U)=r\pi_r(B)$ so with $\alpha=r/(1-r)$, $\pi_r(A)=\alpha \pi_r(U)=\alpha^2 \pi_r(B)$, giving
$$(\pi_r(A),\pi_r(U),\pi_r(B)) = \frac{1}{1+\alpha+\alpha^2} ( \alpha^2, \ \alpha, \ 1).$$
Then,
$$\Phi(r) = \pi_r(A)+\frac{1}{2}\pi_r(U)=\frac{\alpha^2+\alpha/2}{1+\alpha+\alpha^2} = \frac{r^2+r}{2(r^2-r+1)}.$$
We have 
$$\Phi'(r) = \frac{-2r^2+2r+1}{2(r^2-r+1)^2}$$
and 
$$\Phi''(r) = \frac{(-2r+1)(-2r^2+r+2)}{(r^2-r+1)^3}.$$
$\Phi'$ is positive for $r\in[0,1]$ and is an increasing function. Further, $\Phi''$ is positive for $r\in[0,1/2)$, so $\Phi$ is convex over $[0,1/2)$. The Markov chain is symmetric with respect to the simultaneously exchanging $A$ with $B$ and replacing $r$ with $1-r$, from which $\Phi(1-r)=1-\Phi(r)$ (omitted), which implies $\Phi$ is a sigmoid function.

\underline{$M=3$:} $$\Phi(r) = \P(N(X(r))\geq2) = r^3+3r^2(1-r) = 3r^2-2r^3$$
We then have $\Phi'(r)=6r-6r^2=6r(1-r)$ which is positive on $(0,1)$, increasing on $[0,1/2)$ and decreasing on $(1/2,1]$. Hence, $\Phi$ is increasing, convex on $[0,1/2)$ and concave on $(1/2,1]$, thus sigmoid.

\section{Characterization of fixed points of $\Psi_{\CM}$} \label{append:bifur_proof}

Let $\Psi_{\CM}$ be given by \eqref{eq:Psi}. Recall that $r$ is a fixed point of $\Psi_{\mathcal{C_M}}$ if $\Psi_{\mathcal{C_M}}(r)=r$. In this section we characterize the fixed points of $\Psi_{\CM}$ for $M\ge 3$. We already know that $1$ is a fixed point of $\Psi_{\CM}$ for all $\CM$, since $\Phi(1)=1$. Moreover, if $\Phi(1/2)=1/2$ and $\Phi$ is concave on $[1/2,1)$, which holds for $M\ge 3$, then $\Phi(r)>r$ for $1/2<r<1$, so $\Psi_{\CM}$ has no fixed points on $(1/2,1)$. In that case, the following result suffices to characterize its fixed points. It is phrased somewhat more generally, as we only need to assume that $\Phi$ has certain properties. As a result, Theorem \ref{thm:sd} is applicable not only to the ORF \eqref{eq:Phi} for all $M\ge 3$, see Theorem \ref{thm:Phisig}, but also for $\Phi$ from the modified $M=2$ case, see Appendix \ref{append:sig_proof}.
\begin{theorem}\label{thm:sd}
Let $\Phi$ be a $C^1$ (continuously differentiable) function $\Phi$ on $[0,1/2]$ that satisfies $\Phi(0)=0$ and $\Phi(1/2)=1/2$ and is increasing and convex, i.e., both $\Phi$ and $\Phi'$ are positive and increasing on $(0,1/2]$. For $r,\CM \in [0,1/2]$ let
\[\Psi_{\mathcal{C_M}}(r) = \mathcal{C_M} + (1-\mathcal{C_M})\Phi(r)\]
 There exists $\mathcal{C_M}^* \in (0,1/2)$ such that for $0<\mathcal{C_M}<\mathcal{C_M}^*$, $\Psi_{\mathcal{C_M}}$ has two fixed points $0<r_-<r_+<1/2$; for $\mathcal{C_M}=\mathcal{C_M}^*$, $\Psi_{\mathcal{C_M}^*}$ has one fixed point $r^*$; and for $\mathcal{C_M}>\mathcal{C_M}^*$, $\Psi_{\mathcal{C_M}}$ has no fixed points. Moreover, $r_-,r_+$ are $C^1$ functions of $\mathcal{C_M}$ defined on $[0,\mathcal{C_M}^*]$, $r_-$ is increasing, $r_+$ is decreasing, and $\lim_{\CM\to \CM^*} r_-(\CM)=\lim_{\CM\to\CM^*}r_+(\CM)=\Psi(\CM^*)$.
\end{theorem}
\begin{proof}
Let $F(\mathcal{C_M},r)=\Psi_{\mathcal{C_M}}(r)-r$ which is $C^1$ in both $\mathcal{C_M}$ and $r$. Then $r$ is a fixed point of $\Psi_{\mathcal{C_M}}$ iff $F(\mathcal{C_M},r)=0$. Let's show that
\begin{itemize}[noitemsep]
    \item for each $\mathcal{C_M}$, $F(\mathcal{C_M},\cdot)$ is (strictly) convex, and
    \item for each $r$, $F(\cdot,r)$ is (strictly) increasing.
\end{itemize}
For the first point, $\partial_r F(\mathcal{C_M},r) = (1-\mathcal{C_M})\Phi'(r)-1$ is increasing in $r$, since $\Phi'$ is increasing and $0\le \mathcal{C_M} \le 1/2$. For the second point, $\partial_{\mathcal{C_M}} F(\mathcal{C_M},r) = 1-\Phi(r)$ is positive since $0\le \Phi(\cdot)\le 1/2$.

Convexity of $F(\mathcal{C_M},\cdot)$ implies that for any $\mathcal{C_M}$, $F(\mathcal{C_M},\cdot)$ has at most two zeros. By assumption, $F(0,0)=F(0,1/2)=0$ so convexity of $F(0,\cdot)$ implies $\partial_r F(0,0)<0$ and $\partial_r F(0,1/2)>0$. The implicit function theorem implies existence of a unique $C^1$ function $r_-(\mathcal{C_M})$ defined for $\mathcal{C_M}\in [0,\mathcal{C_M}_-)$ for some $\mathcal{C_M}_->0$, satisfying $r_-(0)=0$, $F(\mathcal{C_M},r_-(\mathcal{C_M}))=0$ and

\[r_-'(\mathcal{C_M}) = -\partial_{\mathcal{C_M}} F(\mathcal{C_M},r_-(\mathcal{C_M}))/\partial_r F(\mathcal{C_M},r_-(\mathcal{C_M})).\]

Let $\mathcal{C_M}^*$ be the supremum of values of $\mathcal{C_M}_-$ for which such a function exists. For $\mathcal{C_M}<\mathcal{C_M}^*$, we have $\partial_r F(\mathcal{C_M},r_-(\mathcal{C_M}))\ne 0$, otherwise $r_-'(\mathcal{C_M})$ would be undefined. Moreover, $\partial_r F(\mathcal{C_M}^*,r_-(\mathcal{C_M}^*))=0$, since otherwise, applying the implicit function theorem to $F$ at $(\mathcal{C_M}^*,r_-(\mathcal{C_M}^*))$, the interval $[0,\mathcal{C_M}^*)$ could be made larger, contradicting the definition of $\mathcal{C_M}^*$. Since $\partial_r F(0,0)<0$, by continuity $\partial_r F(\mathcal{C_M},r_-(\mathcal{C_M}))<0$ for $\mathcal{C_M}<\mathcal{C_M}^*$. Further, since $\partial_{\mathcal{C_M}} F>0$, we have $r_-'(\mathcal{C_M})>0$ for $\mathcal{C_M}\in [0,\mathcal{C_M}^*)$. In particular, $\lim_{\mathcal{C_M}\to \mathcal{C_M}^*}r_-(\mathcal{C_M})=:r^*$ exists and by continuity of $F$, we obtain $F(\mathcal{C_M}^*,r^*)=0$. Since $F(\mathcal{C_M}^*,r^*)=\partial_r F(\mathcal{C_M}^*,r^*)=0$ and $F(\mathcal{C_M}^*,\cdot)$ is convex, then $F(\mathcal{C_M}^*,r)>0$ for $r\ne \mathcal{C_M}^*$, i.e., $F(\mathcal{C_M}^*,\cdot)$ has one zero in $[0,1/2]$. Since $\partial_{\mathcal{C_M}} F>0$, we have $F(\mathcal{C_M},r)>0$ for $\mathcal{C_M}>\mathcal{C_M}^*$ and all $r\in[0,1/2]$. An analogous argument gives a function $r_+(\mathcal{C_M})$ defined for $\mathcal{C_M}$ in some interval $[0,a^*)$ satisfying $r_+(0)=1/2$, $F(\mathcal{C_M},r_+(\mathcal{C_M}))=0$, $r_+'(\mathcal{C_M})<0$ and $r_+(\mathcal{C_M})\to w^*$ as $\mathcal{C_M}\to a^*$ for some $w^*$ as well as $F(a^*,w^*)=\partial_r F(a^*,w^*)=0$ and thus $F(\mathcal{C_M},r)>0$ for $\mathcal{C_M}>a^*$ and all $r$. The last part implies $a^*=\mathcal{C_M}^*$ since if $a^*>\mathcal{C_M}^*$ then $F(a^*,w^*)>0$, while if $\mathcal{C_M}^*>a^*$ then $F(\mathcal{C_M}^*,r^*)>0$, both contradictory. Additionally, we have $w^*=r^*$ since $F(\mathcal{C_M}^*,\cdot)$ has one zero.
\end{proof}

\end{appendices}

\end{document}